\documentclass[twocolumn,amsmath,amssymb,superscriptaddress,showpacs]{revtex4}
\usepackage{dcolumn}
\usepackage{bm}
\usepackage{amsthm}

\newtheorem{definition}{Definition}
\newtheorem{theorem}{Theorem}
\newtheorem{lemma}{Lemma}
\newtheorem{proposition}{Proposition}
\newtheorem{condition}{Condition}
\newtheorem{corollary}{Corollary}
\usepackage[pdftex]{graphicx}
\graphicspath{{figures/pdf/}} \DeclareGraphicsExtensions{.pdf}

\def\vec#1{\mathbf{#1}}
\def\gvec#1{{\bm #1}}
\def\sop#1{\mathbf{#1}}
\def\op#1{#1}

\def\ket#1{| #1 \rangle}
\def\bra#1{\langle #1 |}

\def\norm#1{\| #1 \|}

\def\diag{\operatorname{diag}}

\def\dim{\operatorname{dim}}
\def\rank{\operatorname{rank}}
\def\Span{\operatorname{span}}

\def\Tr{\operatorname{Tr}}

\def\A{\mathop{\bf A}\nolimits}
\def\Lb{\mathop{\bf L}\nolimits}
\def\Db{\mathop{\bf D}\nolimits}

\def\D{\mathcal{D}}

\def\H{\mathcal{H}}
\def\I{\mathcal{I}}
\def\L{\mathcal{L}}

\def\S{\mathcal{S}}

\def\ONE{\mathbb{I}}

\def\RR{\mathbb{R}}
\def\DD{\mathfrak{D}}

\def\EE{\mathfrak{E}}

\def\sx{\op{\sigma}_x}
\def\sy{\op{\sigma}_y}
\def\sz{\op{\sigma}_z}

\def\EEss{\EE_{\rm ss}}
\def\EEinv{\EE_{\rm inv}}
\def\EEcc{\EE_{\rm cc}}

\def\DFS{{\rm DFS}}
\def\ss{\rm ss}
\def\lin{\rm lin}
\def\supp{\rm supp}
\begin{document}
\title{Stabilizing Open Quantum Systems by Markovian Reservoir Engineering}
\author{S.~G.~Schirmer}\email{sgs29@cam.ac.uk}
\affiliation{Department of Applied Mathematics and Theoretical
Physics,
             University of Cambridge, Wilberforce Road, Cambridge,
             CB3 0WA, United Kingdom}
\author{Xiaoting Wang}\email{xw233@cam.ac.uk}
\affiliation{Department of Applied Mathematics and Theoretical
Physics,
             University of Cambridge, Wilberforce Road, Cambridge,
             CB3 0WA, United Kingdom}
\date{\today}

\begin{abstract}
We study open quantum systems whose evolution is governed by a
master equation of Kossakowski-Gorini-Sudarshan-Lindblad type and
give a characterization of the convex set of steady states of such
systems based on the generalized Bloch representation.  It is shown
that an isolated steady state of the Bloch equation cannot be a
center, i.e., that the existence of a unique steady state implies
attractivity and global asymptotic stability.  Necessary and
sufficient conditions for the existence of a unique steady state are
derived and applied to different physical models including two- and
four-level atoms, (truncated) harmonic oscillators, and composite
and decomposable systems. It is shown how these criteria could be
exploited in principle for quantum reservoir engineeing via coherent
control and direct feedback to stabilize the system to a desired
steady state.  We also discuss the question of limit points of the
dynamics.  Despite the non-existence of isolated centers, open
quantum systems can have nontrivial invariant sets.  These invariant
sets are center manifolds that arise when the Bloch superoperator
has purely imaginary eigenvalues and are closely related to
decoherence-free subspaces.
\end{abstract}

\maketitle

\section{Introduction}

The dynamics of open quantum systems and especially the possibility
of controlling it have attracted significant interest recently.  One
of the fundamental tasks of interest is the stabilization of quantum
states in the presence of dissipation.  In recent years a large
number of articles have been published on control of closed quantum
systems or, more precisely, on systems that only interact coherently
with a controller, with applications from quantum chemistry to
quantum computing~\cite{qph0602014}.  The essential idea in most of
these articles is open-loop Hamiltonian engineering by applying
control theory and optimization techniques.  Although open-loop
control design is a very important tool for controlling quantum
dynamics, it has limitations. For instance, while open-loop
Hamiltonian engineering can be used to mitigate the effects of
decoherence, e.g., using dynamic decoupling
schemes~\cite{prl.82.2417}, or to implement quantum operations on
logical qubits, protected against errors due to environmental
interactions by a redundant encoding~\cite{NJP11n105032},
Hamiltonian engineering has intrinsic limitations.  One task that is
difficult to achieve using Hamiltonian engineering alone is
stabilization of quantum states.

Alternatively, we can try to engineer open quantum dynamics
described by a Lindblad master
equation~\cite{JMP.76821,CommMathPhys.76119} by changing not only
the Hamiltonian terms but also the dissipative terms.  Various ideas
along these lines have been proposed in several
articles~\cite{pra49p2133,pra.64.063810,pra.71.042309,EuropeanPJD.32.257,
pra.76.010301, pra.78.012334}.  There are two major sources of
dissipative terms in the Lindblad equation: the interaction of the
system with its environment, and measurements we choose to perform
on the system.  Accordingly, we can engineer the open dynamics by
either modifying the system's reservoir or by applying a
carefully-designed quantum measurement.  In this sense, the quantum
Zeno effect is a simple model for reservoir
engineering~\cite{Sudarshan-zeno}.  In addition, the open dynamics
can be modified by feeding the measurement outcome (e.g. the
photocurrent from homodyne detection) back to the controller. This
idea was first proposed in~\cite{pra49p2133}, where a
feedback-modified master equation was derived and it was shown
in~\cite{pra.64.063810} that such direct feedback could be used to
stabilize arbitrary single qubit states with respect to a rotating
frame.  More recently, there have been several attempts to extend
this work to stabilize maximally entangled states using direct
feedback~\cite{pra.64.063810,pra.71.042309,
EuropeanPJD.32.257,pra.76.010301, pra.78.012334}.  The idea of
reservoir engineering can also be used to stabilize the system in
the decoherence-free subspace (DFS)~\cite{DFS}.
In~\cite{Almut-drive}, it is illustrated that $N$ atoms in a cavity
can be entangled and driven into a DFS.
In~\cite{PhysRevA.78.042307}, several interesting physical examples
are presented showing how to design the open dynamics such that the
system can be stabilized in the desired dark state.

Such stabilization problems are a motivation for thorough
investigation of the properties of a Lindblad master equation.
Important questions include, for instance, which states can be
stabilized given a certain general evolution of the system and
certain resources.  There are a number of classical articles
discussing the stationary states and their (asymptotic) stability,
as well as sufficient conditions for the existence of a unique
stationary state~\cite{RepMathPhys.10.76189,
LettMathPhys.2.7733,LettMathPhys.2.7779,CommMathPhys.63.78269,
CommMathPhys.19.7083, CommMathPhys.54.77293}.  More recently, a
detailed analysis of the structure of the Hilbert space with respect
to the Lindblad dynamics was carried out
in~\cite{JPA.41.065201,JPA.41.395303}, implying that all stationary
states are contained in a subspace of the Hilbert space that is
attractive.  Necessary and sufficient conditions for the
attractivity of a subspace or a subsystem have been further
considered in~\cite{IEEE.53.082048,arxiv.0809.0613}.  Nonetheless
there are still important issues that deserve further study.  One is
the issue of asymptotic stability of stationary states.  It is often
assumed that uniqueness implies attractivity of a steady state.
Although this turns out to be true for the Lindblad equation, it
does not follow trivially from the linearity of the master equation,
and a rigorous derivation of this result is therefore desirable, as
is a summary of various sufficient conditions for ensuring
uniqueness of a stationary state. Similarly, linear dynamical
systems can have invariant sets or center manifolds surrounding the
set of steady states.  The existence of such invariant sets usually
precludes converges of the system to a steady state, but criteria
for the existence of non-trivial invariant sets are also of interest
as they are natural decoherence-free subspaces. Finally, many
investigations of the steady states have been based on considering
the dynamics on the Hilbert space of the system, e.g., giving
criteria for the attractivity of a subspace of the Hilbert space.
However, since the steady states are points in the convex set of
positive operators on this Hilbert space, such criteria are not
always useful.  For instance, only systems with steady states at the
boundary of the state space (e.g., pure states) have (non-trivial)
attractive subspaces of the Hilbert space.  While these states may
be of special interest, since the states at the boundary form a set
of measure zero, most systems will have steady states in the
interior.  We may not be able to engineer a steady state at the
boundary, but perhaps we could stabilize a state arbitrarily close
to it, which may be entirely sufficient for practical purposes.
Thus, complete characterization of the steady states requires
considering the set of positive operators on the Hilbert space
rather than the Hilbert space itself.

The purpose of this article is twofold: (i) to further investigate
the properties of the stationary states of the Lindblad dynamics and
the invariant set of the dynamics generated by imaginary
eigenvalues, including the relationship between uniqueness and
asymptotic stability and (ii) to present several sufficient
conditions for the existence of a unique steady state, apply them to
different physical models, and show how these criteria could in
principle be used to stabilize an arbitrary quantum state using
Hamiltonian and reservoir engineering.  In
Sec.~\ref{sec:open-system}, we introduce the Bloch representation of
Lindblad dynamics, which will be used throughout the article.  In
this representation, the spectrum of the dynamics can be easily
derived and stability analysis can be conveniently presented.  In
Sec.~\ref{sec:convex}, we characterize the set of all stationary
states as a convex set generated by a finite number of extremal
points, analyze the properties of the extremal points and give
several sufficient conditions for the uniqueness of the stationary
state.  We also state a theorem that uniqueness implies
attractivity, which is proved in the appendix. In
Sec.~\ref{sec:appli} these conditions are applied to different
systems including two and four-level atoms, the quantum harmonic
oscillator, and composite and decomposable systems, and several
useful results are derived, including: (i) if the Lindblad terms
include the annihilation operator, then the system has a unique
stationary state regardless of the other Lindblad terms or the
Hamiltonian; (ii) for a composite system, if the Lindblad equation
contains dissipation terms corresponding to annihilation operators
for each subsystem, then the stationary state is also unique; (iii)
how any pure or mixed state can be stabilized in principle via
Hamiltonian and reservoir engineering. Finally, in
Sec.~\ref{sec:invariant}, we discuss the invariant set generated by
the eigenstates of the dynamics with purely imaginary eigenvalues,
and its relation to decoherence-free subspaces (DFS), including
examples how to find or design a DFS.

\section{Bloch Representation of Open Quantum System Dynamics}
\label{sec:open-system}

Under certain conditions the evolution of a quantum system
interacting with its environment can be described by a quantum
dynamical semigroup and shown to satisfy a Lindblad master equation
\begin{equation}
 \label{eq:LME}
  \dot\rho(t) = -i[H,\rho(t)] + \L_D\rho(t) \equiv \L \rho(t),
\end{equation}
where $\rho(t)$ is positive unit-trace operator on the system's
Hilbert space $\H$ representing the state of the system, $H$ is a
Hermitian operator on $\H$ representing the Hamiltonian,
$[A,B]=AB-BA$ is the commutator, and
$\L_D\rho(t)=\sum_d\D[V_d]\rho(t)$, where $\op{V}_d$ are operators
on $\H$ and
\begin{equation}
\label{eq:D}
  \D[V_d] \rho(t) = V_d \rho(t) V_d^\dagger
                   - \frac{1}{2}(V_d^\dagger V_d \rho(t)
                   + \rho(t) V_d^\dagger V_d).
\end{equation}
In this work we will consider only open quantum systems governed by
a Lindblad master equation, evolving on a finite-dimensional Hilbert
space $\H\simeq\mathbb{C}^N$.

From a mathematical point of view Eq.~(\ref{eq:LME}) is a complex
matrix differential equation (DE).  To use dynamical systems tools
to study its stationary solutions and the stability, it is desirable
to find a real representation for (\ref{eq:LME}) by choosing an
orthonormal basis $\gvec{\sigma}=\{\sigma_k\}_{k=1}^{N^2}$ for all
Hermitian matrices on $\H$.  Although any orthonormal basis will do,
we shall use the generalized Pauli matrices, suitably normalized,
setting $\sigma_k=\lambda_{rs}$, $k=r+(s-1)N$ and $1\le r<s\le N$,
where
\begin{subequations}
  \label{eq:pauliN}
  \begin{align}
    \lambda_{rs} &= \textstyle \frac{1}{\sqrt{2}}(\ket{r}\bra{s} +
\ket{s}\bra{r}), \\
    \lambda_{sr} &= \textstyle \frac{1}{\sqrt{2}}(-i\ket{r}\bra{s} +i
\ket{s}\bra{r}), \\
    \lambda_{rr}  &= \textstyle \frac{1}{\sqrt{r+r^2}}
          \left(\sum_{k=1}^r \ket{k}\bra{k} - r\ket{r+1}\bra{r+1} \right).
\end{align}
\end{subequations}
The state of the system $\rho$ can then be represented as a real
vector $\vec{r}=(r_k)\in\RR^{N^2}$ of coordinates with respect to
this basis $\{\sigma_k\}$,
\begin{align*}
  \rho=\sum_{k=1}^{N^2}r_k\sigma_k=\sum_{k=1}^{N^2}\Tr(\rho\sigma_k)\sigma_k
\end{align*}
and the Lindblad dynamics (\ref{eq:LME}) rewritten as a real DE:
\begin{equation}
  \label{eq:DE}
  \dot{\vec{r}}=(\Lb+\sum\nolimits_d \Db^{(d)})\vec{r},
\end{equation}
where $\Lb$, $\Db^{(d)}$ are real $N^2\times N^2$ matrices with
entries
\begin{subequations}
\label{eqn:LD}
\begin{align}
   L_{mn}       &= \Tr(i H [\sigma_m,\sigma_n]), \\
   D_{mn}^{(d)} &= \Tr(V_d^\dag \sigma_m V_d \sigma_n)
                   -\frac{1}{2} \Tr(V_d^\dag V_d \{\sigma_m,\sigma_n\}),
\end{align}
\end{subequations}
$\{A,B\}=AB+BA$ being the usual anticommutator.  As
$\sigma_{N^2}=\frac{1}{\sqrt{N}}\ONE$, we have $\dot{r}_N=0$, and
(\ref{eqn:LD}) can be reduced to the dynamics on an
$(N^2-1)$-dimensional subspace,
\begin{equation}
  \label{eq:bloch}
  \dot{\vec{s}}(t) = \A \vec{s}(t) + \vec{c}.
\end{equation}
This is an affine-linear matrix DE in the state vector
$\vec{s}=(r_1,\ldots,r_{N^2-1})^T$. $\A$ is an $(N^2-1)\times
(N^2-1)$ real matrix with $A_{mn}= L_{mn}+ \sum_d D_{mn}^{(d)}$ and
$\vec{c}$ a real column vector with $c_m= L_{mN}+\sum_d
D_{mN}^{(d)}$.  Notice that this essentially is the $N$-dimensional
generalization of the standard Bloch equation for a two-level
system, and we will henceforth refer to $\A$ as the Bloch operator.
The advantage of this representation is that all information of $H$
and $V$ is contained in $\A$ and $\vec{c}$ and it is easy to perform
a stability analysis of the Lindblad dynamics in matrix-vector
form~\cite{PhysRevA.79.052326}.  Defining $\tilde \A=\Lb+\sum_d
\Db^{(d)}$, we have the following relation:
\begin{align*}
\tilde \A=
\begin{pmatrix}
\A &   \sqrt{N}\vec{c} \\
\boldsymbol{0}^T    & 0
\end{pmatrix}.
\end{align*}
Since $\Tr(\rho^2)\le 1$ for any physical state $\rho$, the Bloch
vector $\vec{s}$ must satisfy $\norm{\vec{s}}\le \sqrt{(N-1)/N}$,
i.e. all physical states lie in a ball of radius $R=\sqrt{(N-1)/N}$.
Note that for $N=2$ the embedding into of the physical states into
this ball is surjective, i.e., the set of physical states is the
entire Bloch ball, but this is no longer true for $N>2$.

\section{Characterization of the Stationary States}
\label{sec:convex}

A state $\rho$ is a steady or stationary state of a dynamical system
if $\dot{\rho}=0$.  Steady states are interesting both from a
dynamical systems point of view, as well as for applications such as
stabilizing the system in a desired state.  Let
$\EEss=\{\rho|\dot{\rho}=\L(\rho)=0\}$ be the set of steady states
for the dynamics given by~(\ref{eq:LME}).  As (\ref{eq:LME}) is
linear in $\rho$, $\EEss$ inherits the property of convexity from
the set of all quantum states.  $\EEss$ includes special cases such
as the so-called dark states, which are pure states
$\rho=\ket{\psi}\bra{\psi}$ satisfying $[H,\rho]=\L_D(\rho)=0$.  For
some systems it is easy to see that there are steady states, and
what these are.  For a Hamiltonian system ($\L_D\equiv 0$) it is
obvious from Eq.~(\ref{eq:LME}), for instance, that the steady
states are those that commute with the Hamiltonian, i.e.,
$\EEss=\{\rho:[H,\rho]=0\}$. Similarly, for a system with $H=0$
subject to measurement of the Hermitian observable $M$, the master
equation~(\ref{eq:LME}) can be rewritten as
$\dot\rho=\D[M]\rho=-\frac{1}{2}[M,[M,\rho]]$, and we can show that
$\EEss=\{\rho:[M,\rho]=0\}$.  In general, assuming $\vec{s}_0$ is
the Bloch vector associated with a particular steady state, the set
of steady states $\EEss$ for a system governed by a
LME~(\ref{eq:LME}) can be written as
$\{\vec{s}_0:\A\vec{s}_0+\vec{c}=\vec{0}\}$ in the Bloch
representation.  This is a convex subset of the affine hyperplane
$\EEss^{\lin}=\{\vec{s}_0+\vec{v}\}$ in $\RR^{N^2-1}$, where
$\vec{v}$ satisfies $\A \vec{v}=0$. Moreover, using Brouwer's Fixed
Point Theorem, we can show that the set of steady states $\EEss$ is
always non-empty (see Appendix~\ref{app:ss1}) and we have:

\begin{proposition}
\label{proposition:existence} The Lindblad master
equation~(\ref{eq:LME}) always has a steady state, i.e., the Bloch
equation $\A\vec{s}_0+\vec{c}=\vec{0}$ always has a solution and
$\rank(\A)=\rank(\tilde \A)$, where $\tilde \A$ is the matrix $\A$
horizontally concatenated by the column vector $\vec c$.
\end{proposition}

As any convex set is the convex hull of its extremal points, we
would like to characterize the extremal points of $\EEss$.  A point
in a convex set is called extremal if it cannot be written as a
convex combination of any other points.  See Fig.~\ref{fig:sets} for
illustration of convex sets and extremal points.  To this end, let
$\supp(\rho)$ be the smallest subspace $\S$ of $\H$ such that
$\Pi^\perp \rho \Pi^\perp=0$, where $\Pi$ is the projector onto the
subspace $\S$ and $\Pi^\perp$ is the projector onto the orthogonal
complement of $\S$ in $\H$.

\begin{figure}
\includegraphics[width=\columnwidth]{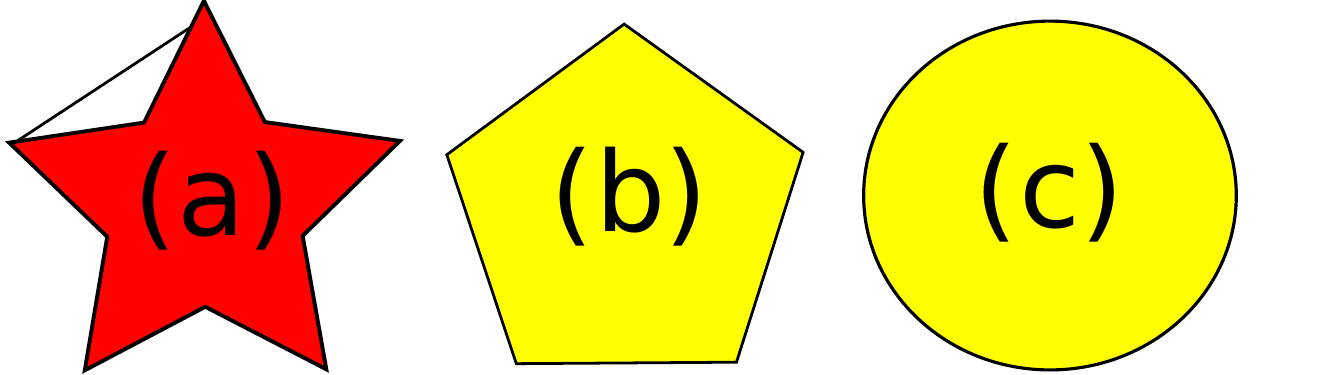} \caption{(Color online) (a)
Non-convex set as a line segment connecting two points in the set is
not contained in the set.  (b) Convex set spanned by five unique
extremal points given by the vertices of the polygon.  (c) Convex
set with infinitely many extremal points comprising the entire
boundary.} \label{fig:sets}
\end{figure}

\begin{proposition}
The steady state of $\EEss$ is extremal if and only if it is the
unique steady state in its support.
\end{proposition}

\begin{proof}
Since any convex set is the convex hull of its extremal points, the
rank of the extremal point is the smallest among its neighboring
points, and the rank of boundary points is smaller than that of
points in the interior. Suppose that besides the extremal steady
state $\rho_0$, there is another steady state $\rho_1$ in the
subspace $\supp(\rho_0)$. Then any state $\rho_2$ which is a convex
combination of $\rho_0$ and $\rho_1$ must also be in
$\supp(\rho_0)$.  However, since $\rho_0$ is an extremal point, the
rank of $\rho_0$, which is equal to the dimension of
$\supp(\rho_0)$, must be lower than the rank of $\rho_2$, which is
impossible.  Conversely, let $\rho_s$ be the unique steady state in
its support.  Suppose it is not an extremal point, which means that
there exist $\rho_1$ and $\rho_2$ with $\rho_s=a\rho_1+(1-a)\rho_2$,
$a>0$. From Lemma~\ref{lemma:app1} in Appendix~\ref{app:extremals},
$\rho_1$ and $\rho_2$ also lie in $\supp(\rho_s)$, a contradiction
to uniqueness of steady states in $\supp(\rho_s)$.
\end{proof}

We call a subspace $\S$ invariant if any dynamical flow with initial
state in $\S$ remains in $\S$.  It has been shown that if
$\rho_{\ss}$ is a steady state then $\supp(\rho_{\ss})$ is
invariant~\cite{JPA.41.395303,arxiv.0809.0613}.  Furthermore,
Proposition~\ref{proposition:existence} shows that any invariant
subspace contains at least one steady state.  Thus, if $\rho_{\ss}$
is an extremal point of $\EEss$ then $\supp(\rho_{\ss})$ is a
\emph{minimal} invariant subspace of the Hilbert space $\H$, i.e.,
there does not exist a proper subspace of $\supp(\rho_k)$ that is
invariant under the dynamics.  It can also be shown that
$\supp(\rho_{\ss})$ is attractive as a subspace of $\H$, and
$\supp(\rho_{\ss})$ has been called a minimal collecting subspace
in~\cite{JPA.41.395303}.

Different extremal steady states generally do not have orthogonal
supports. For example, for a two level-system governed by the
trivial Hamiltonian dynamics $H=0$, $\EEss$ is equal to the convex
set of all states on $\H$, all pure states are extremal points, and
it is easy to see that two arbitrary pure states generally do not
have orthogonal supports.  Just consider the pure states
$\rho_1=\ket{0}\bra{0}$ and
$\rho_2=\frac{1}{2}(\ket{0}+\ket{1})(\bra{0}+\bra{1})$, which are
extremal states but $\supp(\rho_1)\not\perp\supp(\rho_2)$.  However,
in this case there is another extremal steady state
$\rho_3=\ket{1}\bra{1}$ with $\supp(\rho_3) \subset
\supp(\rho_1)+\supp(\rho_2)$ and $\supp(\rho_3) \perp
\supp(\rho_1)$.  In general, given two extremal steady states
$\rho_1$ and $\rho_2$, we have either
$\supp(\rho_1)\perp\supp(\rho_2)$, or there exists another extremal
steady state $\rho_3$ with $\supp(\rho_3)\subset\supp(\rho_1)+
\supp(\rho_2)$ such that $\supp(\rho_1)\perp\supp(\rho_3)$. That is
to say, given an extremal steady state $\rho_1$, if there exist
other steady states, then we can always find another extremal steady
state $\rho_3$ whose support is orthogonal to that of $\rho_1$,
$\supp(\rho_1)\perp\supp(\rho_3)$.  Finally, let $\H_{\ss}$ be the
union of the supports of all steady states $\rho_{\ss}$.  It can be
shown (see, e.g.,~\cite{JPA.41.395303}) that we can choose a finite
number of extremal steady states $\rho_k$ with orthogonal supports,
such that $\H_s=\oplus_{k=1}^K \supp(\rho_k)$.  This decomposition
is generally not unique, however.  In the above example, any two
orthonormal vectors of $\H$ provide a valid decomposition of
$\H_{\ss}=\H$, and no basis is preferable.  Therefore, such a
decomposition of $\H_{\ss}$ is not necessarily physically
meaningful, but it does give the following useful result:

\begin{proposition}
\label{prop:two-steady-states} If a system governed by a
LME~(\ref{eq:LME}) has two steady states, then there exist two
proper orthogonal subspaces of $\H$ that are both invariant.
\end{proposition}

In addition to the characterization of $\EEss$ from the supports of
its extremal points, it is also useful to characterize the steady
states from the structure of the dynamical operators $H$ and $V_d$
in the LME~(\ref{eq:LME}).

\begin{proposition}
\label{thm:ss_bd} If $\rho$ is a steady state at the boundary then
its support $\S=\supp(\rho)$ is an invariant subspace for each of
the Lindblad operators $V_d$.
\end{proposition}

\begin{proof}
A density operator $\rho$ belongs to the boundary of $\DD(\H)$ if it
has zero eigenvalues, i.e., if $\rank(\rho)=N_1<N$.  In this case,
there exists a unitary operator $U$ such that
\begin{equation}
  \label{eq:R11}
  \tilde{\rho} = U \rho U^\dag =
   \begin{bmatrix} R_{11} & R_{12} \\ R_{12}^\dag & R_{22} \end{bmatrix}
\end{equation}
where $R_{11}$ is an $N_1\times N_1$ matrix with full rank, and
$R_{12}$ and $R_{22}$ are $N_2\times N_1$ and $N_2\times N_2$
matrices with zero entries and $N_2=N-N_1=\dim\ker(\rho)$, and
\begin{equation}
  \dot{\tilde{\rho}}(t)
  = -i[\tilde{H},\tilde{\rho}(t)] + \sum_d \D[\tilde{V}_d]\tilde{\rho}(t)
\end{equation}
with $\tilde{H}=UHU^\dag$ and $\tilde{V}_d=UV_dU^\dag$. Partitioning
\begin{equation}
  \label{eq:HVpart}
  \tilde{H} = \begin{bmatrix} H_{11} & H_{12} \\
                              H_{12}^\dag & H_{22}
              \end{bmatrix},
  \quad
  \tilde{V}_d = \begin{bmatrix} V_{11}^{(d)} & V_{12}^{(d)} \\
                                V_{21}^{(d)} & V_{22}^{(d)}
                \end{bmatrix},
\end{equation}
accordingly, it can be verified that a necessary and sufficient
condition for $\rho$ to be a steady state of the system is that
$\dot{R}_{11}=\dot{R}_{12}=\dot{R}_{22}=0$, where
\begin{subequations}
\label{eq:ss_bd}
\begin{align}
 \dot{R}_{11} &= -i[H_{11},R_{11}] + \sum_d\D[V_{11}^{(d)}] R_{11}, \\
 \dot{R}_{12} &= -\frac{1}{2} R_{11}\sum_d (V_{11}^{(d)})^\dag V_{12}^{(d)}+i
R_{11}H_{12}, \\
 \dot{R}_{22} &= \sum_d V_{21}^{(d)} R_{11} (V_{21}^{(d)})^\dag.
\end{align}
\end{subequations}
Since $R_{11}$ is a positive operator with full rank and hence
strictly positive, the third equation requires $V_{21}^{(d)}=0$ for
all $d$.  The second equation is $R_{11}X=0$ for
$X=-\frac{1}{2}\sum_d (V_{11}^{(d)})^\dag V_{12}^{(d)}+iH_{12}$,
which shows that it will be satisfied if and only if the $N_1\times
N_2$ matrix $X$ vanishes identically, which gives the equivalent
conditions
\begin{subequations}
\label{eq:ss_bd2}
\begin{align}
  0 &= -i[H_{11},R_{11}] + \sum_d \D[V_{11}^{(d)}] R_{11},   \\
  0 &= -\frac{1}{2} \sum_d (V_{11}^{(d)})^\dag V_{12}^{(d)}+i H_{12},\\
  0 &= V_{21}^{(d)} \quad \forall d.
\end{align}
\end{subequations}
The last equation implies that if $\rho$ is a steady state at the
boundary then all $V_d$ have a block tridiagonal structure and map
operators defined on $\S=\supp(\rho)$ to operators on $\S$, i.e.,
$\S$ is an invariant subspace for all $V_d$.
\end{proof}

The following theorem (proved in Appendix~\ref{app:no_iso_centers})
shows furthermore that uniqueness implies asymptotic stability:

\begin{theorem}
\label{thm:unique_ss} A steady state of the LME~(\ref{eq:LME}) is
attractive, i.e., all other solutions converge to it, if and only if
it is unique.
\end{theorem}

The fact that only isolated steady states can be attractive
restricts the systems that admit attractive steady states.  In
particular, if there are two (or more) orthogonal subspaces $\H_k$
of the Hilbert space $\H$, which are invariant under the dynamics,
i.e., $\supp \L(\DD(\H_k)) \subset \H_k$ for $k=1,2,\ldots$, then
the dynamics restricted to either invariant subspace must have at
least one steady state on the subspace, and the set of steady states
must contain the convex hull of the steady states on the $\H_k$
subspaces. Thus we have:

\begin{corollary}
\label{cor:no_attractive_state} A system governed by
LME~(\ref{eq:LME}) does \emph{not} have a globally asymptotically
stable equilibrium if there are two (or more) orthogonal subspaces
of the Hilbert space that are invariant under the dynamics.
\end{corollary}

The previous results give several equivalent useful sufficient
conditions to ensure uniqueness of a steady state.

\begin{condition}
\label{cond:1} Given a system governed by a LME~(\ref{eq:LME}) with
an extremal steady state $\rho_{\ss}$, if there is no subspace
orthogonal to $\supp(\rho_{\ss})$ that is invariant under all $V_d$
then $\rho_{\ss}$ is the unique steady state.
\end{condition}

We compare Condition~\ref{cond:1} with Theorem 2
in~\cite{PhysRevA.78.042307}, which asserts that if there exists no
other subspace that is invariant under all $V_d$ orthogonal to the
set of dark states, then the only steady states are the dark states.
To prove that a given dark state is the unique stationay state,
Theorem 2 in~\cite{PhysRevA.78.042307} requires that we show (i)
uniqueness of the dark state, and (ii) that there exists no other
orthogonal invariant subspace.  Since the dark states defined in
\cite{PhysRevA.78.042307} are extremal steady states, Condition
\ref{cond:1} shows that (ii) is actually sufficient in that it
implies uniqueness and hence attractivity of the steady state.

\begin{condition}
\label{cor:unique_ss_interior} If there is no proper subspace of
$\S\subsetneq\H$ that is invariant under all Lindblad generators
$V_d$ then the system has a unique steady state in the
interior~\footnote{Proper subspace means we are excluding the
trivial cases $\S=\{0\}$ and $\S=\H$.}.
\end{condition}

Equation~(\ref{eq:ss_bd2}) also shows that if there are two
orthogonal proper subspaces $\H_1\perp \H_2$ of the Hilbert space
that are invariant under the dynamics, then
$\H=\H_1\oplus\H_2\oplus\H_3$ and there exists a basis such that
\begin{equation*}
\label{eq:ss_bd3}
  H = \begin{bmatrix}
       H_{11} & 0 & H_{13} \\
       0 & H_{22} & 0 \\
       H_{13}^\dag & 0 & H_{33}
      \end{bmatrix}, \quad
  V_d = \begin{bmatrix}
        V_{11}^{(d)} & 0            & V_{13}^{(d)} \\
                   0 & V_{22}^{(d)} & V_{23}^{(d)} \\
                   0 & 0            & V_{33}^{(d)}
        \end{bmatrix}
\end{equation*}
for all $d$, and $iH_{13}-\frac{1}{2}\sum_d (V_{11}^{(d)})^\dag
V_{13}^{(d)}=0$, i.e., in particular both subspaces are $V_d$
invariant for all $V_d$.  Hence, if there are no two orthogonal
proper subspaces of $\H$ that are simultaneously $V_d$ invariant for
all $V_d$, then the system does not admit orthogonal proper
subspaces that are invariant under the dynamics. Thus we have:

\begin{condition}
\label{cor:unique_ss} If there do not exist two orthogonal proper
subspaces of $\H$ that are simultaneously $V_d$ invariant for all
$V_d$ then the system has a unique fixed point, either at the
boundary or in the interior.
\end{condition}

The following applications show that these conditions are very
useful to show attractivity of a steady state.

\section{Applications}
\label{sec:appli}

\subsection{Two and Four-level Atoms}

Let us start with the simplest example, a two-level atom governed by
the Lindblad master equation
\begin{align*}
\dot \rho = -i\Omega[\sigma_x,\rho]+\D[\sigma]\rho
\end{align*}
with $\sigma=\ket{0}\bra{1}$.  This model describes a two-level atom
subject to spontaneous emission, or a two-level atom interacting
with a heavily damped cavity field after adiabatically eliminating
the cavity mode.  Noting that the Lindblad operator $\sigma$
corresponds to a Jordan matrix $J_0(2)$, the previous results
guarantee that this system has a unique (attractive) steady state.
More interestingly, the previous results still guarantee the
existence of a unique steady state if the atom is damped by a bath
of harmonic oscillators
\begin{align*}
 \dot\rho=[-iH,\rho]-\frac{\Gamma}{2}\bar n
 \D[\sigma^\dag]\rho-\frac{\Gamma}{2}(\bar n+1) \D[\sigma]\rho,
\end{align*}
where $\bar n=(e^{\hbar \omega/k_BT}-1)^{-1}$ is the average photon
number.  It suffices that one of the Lindblad term $\D[\sigma]\rho$
corresponds to an indecomposable Jordan matrix.  In this simple case
we can also infer the uniqueness of the steady state directly from
the Bloch representation.  We can decompose the Bloch matrix
$\A=\A_H+\A_D$ into an antisymmetric matrix $\A_H$ corresponding to
the Hamiltonian part of the evolution and a diagonal and
negative-definite matrix $\A_D$.  Since
$\vec{s}^T\A\vec{s}=\vec{s}^T\A_D\vec{s}<0$ for any $\vec{s}\neq 0$,
it follows that $\A_D$ is invertible and the Bloch equation
$\dot{\vec{s}}=\A\vec{s}+\vec{c}$ has a unique attractive stationary
state.

On the other hand, if the atom is subjected to a continuous weak
measurement such as $\dot \rho=\D[\sigma_z] \rho$ then we can easily
verify that the pure states $\ket{0}$ and $\ket{1}$ are steady
states. Hence, there are infinitely many steady states given by the
convex hull of these extremal points, $\rho_{\ss}=\alpha
\ket{0}\bra{0} +(1-\alpha)\ket{1}\bra{1}$ with $0\le\alpha \le 1$.
Of course, this is the well-known case of a depolarizing channel,
which contracts the entire Bloch ball to the $z$ axis, which is the
measurement axis.

In the previous examples uniqueness of the steady state followed
from similarity of at least one Lindblad operator $V$ to an
(indecomposable) Jordan matrix.  When $V$ is decomposable then the
last example shows that the system can have infinitely many steady
states, but similarity of a Lindblad operator to an indecomposable
Jordan matrix is only a sufficient condition, i.e., it is not
necessary for the existence of a unique steady state. If $V$ has two
or more Jordan blocks, for example, then each Jordan block defines
an invariant subspace, but provided these subspaces are not
orthogonal to each other, Condition 3 still applies, ensuring the
uniqueness of the steady state.

For instance, a system governed by a LME $\dot\rho = \D[V]\rho$ with
$V=S^{-1}JS$, $J=J_0(2)\oplus J_1(2)$ and
\begin{equation*}
S = \begin{pmatrix}
           1 & 0 & 0 & 0\\
           0 & 1 & 1 & 0\\
           0 & 0 & 0 & 1\\
          1 & 0 & 1 & 0
          \end{pmatrix}
\end{equation*}
has a unique steady state because, although $V$ has two eigenvalues
$0$ and $1$ and two proper eigenvectors, the respective eigenspaces
are \emph{not} orthogonal and there are no two orthogonal subspaces
that are invariant under $V$.  Perhaps more interestingly, for a
system with a nontrivial Hamiltonian, e.g.,
$\dot\rho=-i[H,\rho]+\D[V]\rho$, uniqueness of the steady state can
often be guaranteed even if $V$ has two (or more) orthogonal
invariant subspaces, if $H$ suitably mixes the invariant subspaces.

Consider a four-level system with energy levels as illustrated in
Fig.~\ref{fig:four_level_atom} and spontaneous emission rates
$\gamma_{34}$, $\gamma_{23}$ and $\gamma_{12}$ satisfying
$\gamma_{34},\gamma_{12}\ge \gamma_{23}$.  This is a simple model
for a laser.  To derive stimulated emission we require population
inversion, a cavity and a gain medium composed of many atoms.  For
simplicity, we only consider one atom and try to describe the
dynamics in the time scale such that the spontaneous decay
$3\rightarrow 2$ can be neglected. On this scale the Hamiltonian
optical-pumping term $H$ and the spontaneous decay term are:
\begin{align*}
    H &=\alpha(\ket{1}\bra{4}+\ket{4}\bra{1}),\\
  V_1 &=\gamma_{34}\ket{3}\bra{4},\\
  V_2 &=\gamma_{12}\ket{1}\bra{2}.
\end{align*}
There are two invariant subspaces under $V_1$ and $V_2$:
$\H_1=\Span\{\ket{1},\ket{2}\}$ and $\H_2=\Span\{\ket{3},\ket{4}\}$.
Hence, when $\alpha=0$, we have two metastable states $\ket{2}$ and
$\ket{3}$ in addition to the ground state $\ket{1}$, which is a
steady state.  However, for $\alpha\ne0$ the pumping Hamiltonian $H$
mixes up those two invariant subspaces, and through calculation we
can easily find the unique steady state:
$\rho_{\ss}=\ket{3}\bra{3}$. Thus, on the time scales considered,
population inversion between states $\ket{3}$ and $\ket{2}$ can be
realized, but eventually spontaneous emission from $\ket{3}$ to
$\ket{2}$ will kick in, resulting in the stimulated emission
characteristic of a laser. (Of course, this is only the first stage
of the whole process and it is far from the threshold of the laser.)

\begin{figure}
\center
\includegraphics[width=0.5\columnwidth]{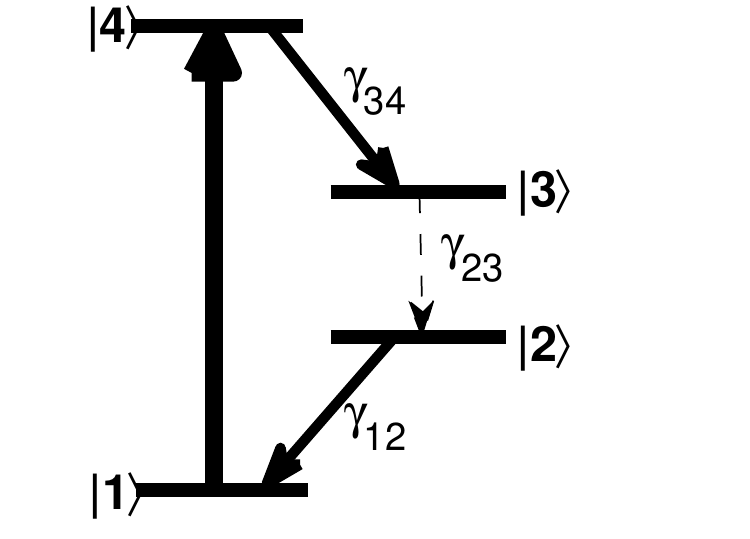}
\caption{Schematic plot of the four energy levels of one atom in a
prototype system for a laser.  The atom is pumped by an external
field. The spontaneous decay rates satisfy
$\gamma_{34},\gamma_{12}\ge \gamma_{23}$.  On the time scale when
decay from level $3$ to $2$ can be ignored
$\rho_{\ss}=\ket{3}\bra{3}$ is the unique steady state of the
system, realizing the population inversion.}
\label{fig:four_level_atom}
\end{figure}

This is just one example of optical pumping, a technique widely used
for state preparation in quantum optics.  Although the principle of
optical pumping is easy to understand intuitively for simple systems
in that population cannot accumulate in energy levels being pumped,
forcing the population to accummulate in states state not being
pumped and not decaying to other states, it can be difficult to
intuitively understand the dynamics in less straightforward cases.
For example, what would happen if we applied an additional laser
field coupling $\ket{3}$ and $\ket{4}$.  Would the system still have
a unique steady state?  If so, what is the steady state?  These
questions are not easy to answer based on intuition, but we can very
easily answer them using the mathematical formalism developed,
especially the Bloch equation.  In fact, we easily verify that the
system
\begin{equation*}
  \dot\rho(t) = -i[H,\rho] + \gamma_{34}\D[\ket{3}\bra{4}] \rho +
                             \gamma_{12}\D[\ket{1}\bra{2}] \rho
\end{equation*}
with $H=\alpha (\ket{1}\bra{4}+\ket{4}\bra{1}) +\beta
(\ket{3}\bra{4}+\ket{4}\bra{3})$ has a unique steady state
\begin{equation*}
  \rho_{\ss} = \frac{1}{\alpha^2+\beta^2}
   \begin{bmatrix}
                \beta^2 & 0 & -\alpha\beta & 0 \\
                0 & 0 & 0 & 0 \\
                -\alpha\beta & 0 & \alpha^2 & 0 \\
                0 & 0 & 0 & 0
   \end{bmatrix}
\end{equation*}
independent of $\gamma_{12}$ and $\gamma_{34}$, provided
$\gamma_{12},\gamma_{34}\neq 0$.  For $\beta=0$ this state becomes
$\ket{3}\bra{3}$, as intuition suggests.

\subsection{Quantum Harmonic Oscillator}

The harmonic oscillator plays an important role as a model for a
wide range of physical systems from photon fields in cavities, to
nano-mechanical oscillators, to bosons in the Bose-Hubbard model for
cold atoms in optical lattices.  Although strictly speaking the
harmonic oscillator is defined on an infinite-dimensional Hilbert
space, the dynamics can often be restricted to a finite-dimensional
subspace.  For many interesting quantum processes the average energy
of the system is finite and we can truncate the number of Fock
states $N_{\rm max}$ from $\infty$ to a large but finite number.  In
many quantum optics experiments, for example, the intracavity field
contains only a few photons, or has a number of photons in some
finite range if it is driven by a field with limited intensity.  In
such cases the truncated harmonic oscillator is a good model for the
underlying physical system provided $N_{\rm max}$ is large enough,
and we can apply the previous results about stationary solutions and
asymptotic stability.

Consider a harmonic oscillator with $H_0=\hbar\omega c^\dag c$ where
$c$ is the annihilation operator of the system, which on the
truncated Hilbert space with $N_{\rm max}=N$, takes the form
\begin{equation}
  c\propto \sum_{n=0}^{N-1}\sqrt{n+1}\ket{n}\bra{n+1}.
\end{equation}
If there is a Lindblad term of the form $\D[c]\rho$ then we can
infer from the previous analysis that the system has a unique and
hence asymptotically stable steady state, regardless of whatever
Hamiltonian control or interaction terms or other Lindblad terms are
present.  To see this note that the matrix representation of $c$ is
mathematically similar to the Jordan matrix
\begin{equation}
   J_0(N) =\sum_{n=0}^{N-1}\ket{n}\bra{n+1}.
\end{equation}
It is easy to verify that $J$ has a sole proper eigenvector whose
generalized eigenspace is all of $\H$ and thus does not admit two
orthogonal proper invariant subspaces.  Hence we can conclude from
Condition 3 that for \emph{any} dynamics governed by a
LME~(\ref{eq:LME}) with a dissipation term $\D[c]\rho$, there is
always a unique stationary solution to which any initial state will
converge. In general, if (\ref{eq:LME}) contains a Lindblad term
$\D[V]\rho$ with $V$ similar to a Jordan matrix
$J_\alpha(N)=\alpha\ONE_N+J_0(N)$, then (\ref{eq:LME}) always has a
unique stationary state, no matter what the other terms are.  For
example, the Lindblad equation for a damped cavity driven by a
classical coherent field $\alpha$ is
\begin{align*}
\dot\rho = -\frac{1}{2}[\alpha^* c-\alpha c^\dag,\rho]+\D[c]\rho
         = \D[\alpha\ONE_N+c]\rho,
\end{align*}
showing that the system has a unique steady state.  For $N=4$ the
steady state is
\begin{equation*}
  \rho_{\ss} = \frac{1}{C}
  \begin{pmatrix}
  1+\alpha^2A   & -\alpha A    &  \alpha^2 B & -\alpha^3 \\
   -\alpha A    &  \alpha^2 A  & -\alpha^3 B & \alpha^4 \\
   \alpha^2 B   &  -\alpha^3 B &  \alpha^4 B & -\alpha^5 \\
   -\alpha^3    &  \alpha^4    & -\alpha^5   &  \alpha^6
  \end{pmatrix}
\end{equation*}
with $\alpha$ real, $A=\alpha^2B+1$, $B=\alpha^2+1$ and
$C=4\alpha^6+3\alpha^4+2\alpha^2+1$.  When $\alpha=0$, i.e. there is
no driving field, we get $\rho_{\ss}$ is the ground (vacuum) state,
as one would expect for a damped cavity, while for a nonzero driving
field we stabilize a mixed state in the interior.

\subsection{Composite Systems}

Many physical systems are composed of subsystems, each interacting
with its environment, inducing dissipation.  For example, consider
$N$ two-level atoms in a damped cavity driven by a coherent external
field. Assuming the atom-atom and atom-cavity interactions are not
too strong, and the main sources of dissipation are independent
decay of atoms and the cavity mode, respectively, we obtain the
Lindblad terms $\D[\sigma_n]$, $n=1,\ldots,N$, and $\D[c]$ in the
LME (\ref{eq:LME}), where $\sigma_n$ is the decay operator
$\sigma=\ket{0}\bra{1}$ for the $n$th atom and $c$ is the
annihilation operator of the cavity. Simulations suggest systems of
this type always have a unique steady state, and this can be
rigorously shown using the sufficient conditions derived.

A composite quantum system whose evolution is governed by a LME
containing terms involving annihilation operators for each subsystem
has a unique steady state, regardless of the Hamiltonian and any
other Lindblad terms that may be present.  This property can be
inferred from Condition~3.  Assume the full system is composed of
$K$ subsystems with Lindblad terms $\D[\sigma_k]\rho$,
$k=1,\ldots,K$ and let $\H_I$ be an invariant subspace for all
$\sigma_k$.  Then $\H_I$ must contain the ground state $\ket{{\bf
0}}=\ket{0}^{\otimes K}$ of the composite system as
$\sigma_k\ket{0}=0$ for all $k$.  Hence, any simultaneously
$\sigma_k$-invariant subspace must contain the state $\ket{{\bf 0}}$
and there cannot exist two orthogonal proper subspaces of $\H$ that
are invariant under all $\sigma_k$.  By Condition~3, the system has
a unique steady state.

Thus, a system of $N$ atoms in a damped cavity subject to a Lindblad
master equation
\begin{equation*}
  \dot\rho(t) = -i[H,\rho(t)] + \gamma\D[c]\rho
                 + \sum_{n=1}^N \gamma_n \D[\sigma_n]\rho
\end{equation*}
has a unique steady state, regardless of the Hamiltonian $H$.  The
steady state need \emph{not} be $\ket{{\bf 0}}$, however.  In
general, this will only be the case if $\ket{{\bf 0}}$ is an
eigenstate of $H$. Similarly, the presence of the two dissipation
terms $\D[\sigma_k]$ in the LME for the two-atom model in
\cite{pra.71.042309}
\begin{equation*}
  \dot\rho(t) = -i[J+J^\dag,\rho(t)] + \gamma\D[J]\rho
  + \sum_{k=1,2}\gamma_k\D[\sigma_k]\rho
\end{equation*}
with $J=\sigma_1+\sigma_2$ ensures that there is a unique steady
state provided $\gamma_k>0$.  This is no longer the case for
$\gamma_k=0$.  In particular, in the regime where
$\gamma_1,\gamma_2\ll \gamma$ and the last two terms can be
neglected as in \cite{pra.71.042309}, the reduced dynamics no longer
has a unique steady state.

\subsection{Decomposable Systems}

\begin{figure}
\includegraphics[width=0.99\columnwidth]{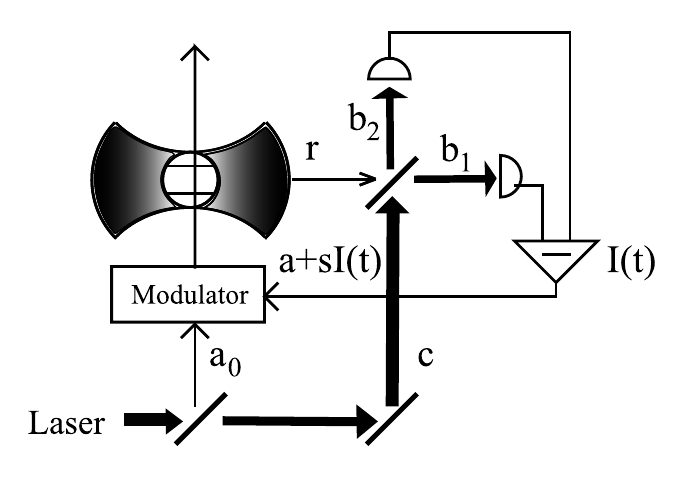}
\caption{Schematic diagram of an atomic system under direct quantum
feedback control. A laser beam is split to generate the local
oscillator $c$ for homodyne detection as well as the driving field
$a_0$ which is then modulated by the feedback photocurrent $I(t)$ to
derive the final field $a+sI(t)$.} \label{fig:homo}
\end{figure}

A system is decomposable if there exists a decomposition of the
Hilbert space $\H=\oplus_{m=1}^M \H_m$ such that
$\dot\rho=\oplus_{m=1}^M \dot\rho^{(m)}$ for any $\rho(0)=
\oplus_{m=1}^M \rho^{(m)}(0)$ where $\rho^{(m)}(0)$ is an
(unnormalized) density operator on $\H_m$. Decomposable systems
cannot have asymptotically stable (attractive) steady states by
Corollary~\ref{cor:no_attractive_state}.

One class of systems that are always decomposable and hence never
admit attractive steady states, are systems governed by a
LME~(\ref{eq:LME}) with a single Lindblad operator $V$ that is
normal, i.e., $[V,V^\dag]=0$, and commutes with the Hamiltonian.
This is easy to see. Normal operators are diagonalizable, i.e.,
there exists a unitary operator $U$ such that $UVU^\dag=D$ with $D$
diagonal, and since $[H,V]=0$, we can choose $U$ such that it also
diagonalizes $H$.  Thus the system is fully decomposable, and it is
easy to see in this case that every joint eigenstate of $H$ and $V$
is a steady state, and therefore there exists a steady-state
manifold spanned by the convex hull of the projectors onto the joint
eigenstates of $H$ and $V$.  In the absence of degenerate
eigenvalues this manifold is exactly the $N-1$ dimensional subspace
of $\DD(\H)$ consisting of operators diagonal in the joint
eigenbasis of $H$ and $V$.

A more interesting example of a physical system that is
decomposable, and thus does not admit an attractive steady state, is
a system of $n$ indistinguishable two-level atoms in a cavity
subject to collective decay, and possibly collective control of the
atoms as well as collective homodyne detection of photons emitted
from the cavity, as illustrated in Fig.~\ref{fig:homo}.  Let
$\sigma=\ket{0}\bra{1}$ be the single-qubit annihilation operator
and define the single-qubit Pauli operators
$\sx=\sigma+\sigma^\dag$, $\sy=i(\sigma-\sigma^\dag)$ and
$\sz=-2[\sx,\sy]$.  Choosing the collective measurement operator
\begin{equation*}
  M=\sum_{\ell=1}^n\sigma^{(\ell)}
\end{equation*}
$\sigma^{(\ell)}$ being the $n$-fold tensor product whose $\ell$th
factor is $\sigma$, all others being the identity $\ONE_2$, and the
collective local control and feedback Hamiltonians
\begin{equation*}
  H_c=u_x J_x + u_y J_y + u_z J_z, \quad F=\lambda H_c,
\end{equation*}
where $J_a=\sum_{\ell=1}^n \sigma_a^{(\ell)}$ for $a\in\{x,y,z\}$,
the evolution of the system is governed by the feedback-modified
Lindblad master equation~\cite{pra49p2133}
\begin{equation*}
  \dot{\rho}(t) = -i[H_0+H_c+M^\dag F+F M,\rho] + \D[M-iF]\rho,
\end{equation*}
assuming local decay of the atoms is negligible.  It is easy to see
from the master equation above that the system decomposes into
eigenspaces of the (angular momentum) operator
\begin{equation*}
  J=J_x^2+J_y^2+J_z^2,
\end{equation*}
i.e., both the measurement operator $M$ and the control and feedback
Hamiltonians $H_c$ and $F$ (and hence $M^\dag F+F M$) can be written
in block-diagonal form with blocks determined by the eigenspaces of
$J$. Therefore, the system is decomposable and we cannot stabilize
\emph{any} state, no matter how we choose
$\vec{u}=(u_x,u_y,u_z,\lambda)$.  For $n=2$ this system was studied
in~\cite{pra71n042309} in the context of maximizing entanglement of
a steady state on the $J=1$ subspace using feedback, although the
question of stability of the steady states was not considered.
Although the system does not admit an attractive steady state in the
whole space, we can verify that $\EEss$ contains a line segment of
steady states that intersects both the $J=0$ and $J=1$ subspaces in
a unique state.  Thus $J=1$ subspace has a unique steady state
determined by $\vec{u}$, to which all solutions with initial states
in this subspace converge.

\subsection{Feedback Stabilization}

An interesting possible application of the criteria for the
existence of unique, attractive steady states is the possibility of
engineering the dynamics such that the system has a desired
attractive steady state by means of coherent control, measurements
and feedback.  An special case of interest here is direct feedback.
Systems subject to direct feedback as in the previous example, can
be described by a simple \emph{feedback-modified} master
equation~\cite{pra49p2133}:
\begin{equation}
 \label{eq:FME}
 \dot\rho(t) = -i[H,\rho(t)] + \D[M-iF]\rho(t),
\end{equation}
where $H=H_0+H_c+\frac{1}{2}(M^\dag F+FM)$ is composed of a fixed
internal Hamiltonian $H_0$, a control Hamiltonian $H_c$ and a
feedback correction term $\frac{1}{2}(M^\dag F+FM)$.  This master
equation is of Lindblad form, and hence all of the previous results
are directly applicable.  Setting
\begin{subequations}
\begin{align}
 V   &= M-iF,\\
 M   &= V+V^\dag, \\
 F   &= i(V-V^\dag), \\
 H_c &= \textstyle H-H_0-\frac{1}{2}(M^\dag F+FM)
\end{align}
\end{subequations}
we see immediately that if the control and feedback Hamiltonian,
$H_c$ and $F$, and the measurement operator $M$ are allowed to be
arbitrary Hermitian operators, then we can generate \emph{any}
Lindblad dynamics. This is also true for a non-Hermitian measurement
operator $M$ as arises, e.g., for homodyne detection, since the
anti-Hermitian part of $M$ can always be canceled by the effect of
the feedback Hamiltonian $F$ in $\D[M-iF]$.  Given this level of
control, it is not difficult to show that we can in principle render
\emph{any} given target state $\rho_{\ss}$, pure or mixed, globally
asymptotically stable by choosing appropriate $H$ and $V$ or,
equivalently, by choosing appropriate $H_c$, $F$ and $M$.

To see how to do accomplish this in principle, let us first consider
the generic case of a target state $\rho_{\ss}$ is in the interior
of the convex set of the states with $\rank(\rho_{\ss})=N$. A
necessary and sufficient condition for $\rho_{\ss}$ to be an
attractive steady state is
\begin{itemize}
\item[(i)]  $-i[H,\rho_{\ss}]+\D[V]\rho_{\ss}=0$ and
\item[(ii)] no (proper) subspace of $\H$ is invariant under (\ref{eq:LME}).
\end{itemize}
The first condition ensures that $\rho_{\ss}$ is a steady state, and
the latter ensures that it is the only steady state in the interior
by Corollary~\ref{cor:unique_ss_interior}.  It is easy to see that
choosing $V$ and $H$ such that
\begin{equation}
  \label{eq:ss_in}
  V = U\rho_{\ss}^{-1/2}, \quad [H,\rho_{\ss}] =0
\end{equation}
where $U$ is unitary, ensures that (i) is satisfied as
\begin{align*}
  \D[V]\rho_{\ss}&= U\rho_{\ss}^{-1/2}\rho_{\ss}\rho_{\ss}^{-1/2}U^\dag
                  -\frac{1}{2}\{\rho_{\ss}^{-1},\ket{n}\bra{n}\}\\
                 &= UU^\dag - \frac{1}{2} \{\rho_{\ss}^{-1},\rho_{\ss}\}
                  = \ONE - \ONE =0.
\end{align*}
To satisfy (ii) we must choose $U$ such that $V$ has no orthogonal
invariant subspaces, or equivalently $H$ mixes up any two orthogonal
invariant subspaces $V$ may have.  If $\rho_{\ss}=\sum_k w_k\Pi_k$,
where $\Pi_k$ is the projector onto the $k$th eigenspace then the
invariance condition implies that $U$ must not commute with any of
the projection operators $\Pi_k$, or any partial sum of $\Pi_k$ such
as $\Pi_1+\Pi_2$.  To see this, suppose $U$ commutes with
$\Pi_n=\ket{n}\bra{n}$, a projector onto an eigenspace of
$\rho_{\ss}$. Then $\ket{n}$ is a simultaneous eigenstate of $U$ and
$\rho_{\ss}$ with $U\ket{n}=e^{i\phi}\ket{n}$ and
$\rho\ket{n}=\alpha\ket{n}$, where $\alpha$ must be real and
positive as $\rho_{\ss}$ is a positive operator, and we have
$[H,\ket{n}\bra{n}]= 0$, $\{\rho_{\ss}^{-1},\ket{n}\bra{n}\}=
2\alpha^{-1} \ket{n}\bra{n}$, $\ket{n}$ is an eigenstate of $V$
\begin{equation*}
    V\ket{n} = U\rho^{-1/2}\ket{n}
             = U \alpha^{-1/2} \ket{n}
             = \alpha^{-1/2} e^{i\phi}\ket{n},
\end{equation*}
and thus $V\ket{n}\bra{n}V^\dag=\alpha^{-1}\ket{n}\bra{n}$ and
$\D[V]\ket{n}\bra{n}=0$, i.e., $\ket{n}\bra{n}$ is a steady state of
the system at the boundary.  Hence, the steady state is not unique,
and $\rho_{\ss}$ cannot be attractive.  In practice almost any
randomly chosen unitary matrix $U$ such as $U=\exp(i
(X+X^\dagger))$, where $X$ is a random matrix, will satisfy the
above condition, and given a candidate $U$ it is easy to check if it
is suitable by calculating the eigenvalues of the superoperator $\A$
in~(\ref{eq:bloch}).  Of course, choosing $H$ and $V$ of the
form~(\ref{eq:ss_in}) is just one of many possible choices for
condition (i) to hold.  It is possible to find other suitable sets
of operators $(H,V)$ in terms of $(H_c,F,M)$ when the class of
practically realizable control and feedback operators or
measurements is restricted.  For example, we can easily verify that
$\rho_{\ss}$ is the unique attractive steady state of a two-level
system governed by the LME~(\ref{eq:LME}) with
\begin{align*}
  H = \begin{bmatrix} 0 & 1\\ 1 & 0 \end{bmatrix},
  V = \begin{bmatrix} 1 & (\sqrt{3}-1)i\\ (3-\sqrt{3})i & 1\end{bmatrix},
  \rho_{\ss}=\begin{bmatrix} \frac{1}{4} & 0\\0 & \frac{3}{4}\end{bmatrix},
\end{align*}
even though $H$ and $V$ do not satisfy~(\ref{eq:ss_in}).  Thus,
there are generally many possible choices for the control,
measurement, and feedback operators that render a particular state
in the interior asymptotically stable.

If the target state $\rho_{\ss}$ is in the boundary of the convex
set of the states, i.e., $\rank(\rho_{\ss})<N$, then the proof of
Proposition~\ref{thm:ss_bd} shows that we must have
\begin{subequations}
\label{eq:ss_bd4}
\begin{align}
  0 &= -i[H_{11},R_{11}] + \D[V_{11}] R_{11},   \\
  0 &= -\frac{1}{2} V_{11}^\dag V_{12}+i H_{12},\\
  0 &= V_{21}.
\end{align}
\end{subequations}
with $V_{k\ell}$ and $H_{k\ell}$ defined as in
Eq.~(\ref{eq:HVpart}), to ensure that $\rho$ is a steady state.  To
ensure uniqueness we must further ensure that there are no other
steady states.  This means, by Corollary~\ref{cor:unique_ss}, that
(a) we must choose $H_{11}$ and $V_{11}$ such that $R_{11}$ is the
\emph{unique} solution of (\ref{eq:ss_bd4}a), and thus no subspace
of $\S=\supp(\rho_{\ss})$ is invariant, and (b) we must choose the
remaining operators $H_{12}$ and $V_{12}$ and $V_{22}$ such that
(\ref{eq:ss_bd4}b) is satisfied and no subspace of $\S^\perp$ is
invariant, because if such a subspace $\S_2$ exists, then $\S_1$ and
$\S_2$ will be two proper orthogonal invariant subspaces and
$\rho_{\ss}$ will not be attractive.

One way to construct such a solution is by choosing $H_{11}$ such
that $[H_{11},R_{11}]=0$ and setting $V_{11}=U_{11}R_{11}^{-1/2}$,
where $U_{11}$ is a suitable unitary operator defined on $\S$ as
discussed in the previous section.  Then we choose $V_{22}$ such
that no proper subspace of $\S^{\perp}$ is invariant.  Finally, we
must choose $V_{12}$ and $H_{12}$ such that (\ref{eq:ss_bd4}b) is
satisfied and $\S^\perp$ is itself not invariant.  Although these
constraints appear quite strict, in practice there are usually many
solutions.

For example, suppose we want to stabilize the rank-3 mixed state
$\rho_{\ss}=\frac{1}{8}\diag(1,3,4,0)$ at the boundary.  Then we
partition $\rho$, $V$ and $H$ as above, setting $V_{11}=U R_{11}$
with $R_{11}=\frac{1}{8}\diag(1,3,4)$ and $U$ a suitable unitary
matrix such as
\begin{equation*}
  U = \begin{bmatrix} 0 & 1 & 0 \\ 0 & 0 & 1 \\ 1 & 0 & 0 \end{bmatrix}.
\end{equation*}
Then we choose $H_{11}$ such that $[H_{11},R_{11}]=0$, e.g., we
could set $H_{11}=V_{11}^\dag V_{11}$, a choice, which ensures that
$\rho_{\ss}$ is the unique steady state on the subspace
$\S=\supp(\rho_{\ss})$.  Next we choose $V_{12}$ such that
$\S^\perp$ is not an invariant subspace.  Any choice other than
$V_{12}=(0,0,0)$ will do in this case, e.g., set $V_{12}=(1,0,0)$.
Finally, we set $H_{12}=-\frac{i}{2}V_{11}^\dag V_{12}$,
$V_{21}=(0,0,0)^T$ and $V_{22}\neq 0$ to ensure that $\rho_{\ss}$ is
the unique globally asymptotically stable state.

Note that the Hamiltonian, which was \emph{not} crucial for
stabilizing a state in the interior and could have been set to
$H=0$, does affect our ability to stabilize states in the boundary.
We can stabilize a mixed state in the boundary only if $H_{12}\ne0$.
If $H_{12}=0$ then Eq.~(\ref{eq:ss_bd4}b) implies $V_{11}^\dag
V_{12}=0$, and there are two possbilities.  If $V_{12}\ne0$ but
$V_{11}$ has a zero eigenvalue, then the system restricted to the
subspace $\S$ has a pure state at the boundary and thus $R_{11}$
cannot be the unique attractive steady state on $\S$. Alternatively,
if $V_{12}=0$ then $V$ is decomposable with two orthogonal invariant
subspaces $\S$ and $\S^\perp$, and $\rho_{\ss}$ cannot be attractive
either, consistent with what was observed in~\cite{JPA41n065201}.

Target states at the boundary include pure states.  $\S$ in this
case is a one-dimensional subspace of $\H$, and
Eq.~(\ref{eq:ss_bd4}a) is trivially satisfied as $H_{11}$ and
$V_{11}$ have rank 1, and the crucial task is to find a solution to
Eq.~(\ref{eq:ss_bd4}b) such that no subspace of $\S^\perp$ is
invariant.  If $H_{12}=0$ then this is possible only if $V_{11}=0$
and thus if $V$ has a zero eigenvalue, as was observed
in~\cite{JPA41n065201}, but again, if $H\neq 0$ then there are many
choices for $H$ and $V_{12}$, $V_{22}$ that stabilize a desired pure
state.  For example, we can easily check that the pure state
$\rho_{\ss}=\ket{1}\bra{1}$ is a steady state of the system
$\dot\rho=-i[H,\rho]+\D[V]\rho$ if $V$ is the irreducible Jordan
matrix $J_a(N)$ with eigenvalue $a$ and
$H_{12}=-\frac{i}{2}a^*\ket{1}\bra{2}$.

\section{Invariant Set of Dynamics, Decoherence-free subspaces}
\label{sec:invariant}

Having characterized the set of steady states, the question is
whether the system always converges to one of these equilibria.  The
previous sections show that this is the case if the system has a
unique steady state, as uniqueness implies asymptotic stability.  In
general, however, this is clearly not the case for a linear
dynamical system.  Rather, all solutions converge to a center
manifold $\EEinv$, which is an invariant set of the dynamics,
consisting of both steady states and limit
cycles~\cite{94Glendinning}.  Although we have seen that the
Lindblad master equation~(\ref{eq:LME}) does not admit
\emph{isolated} centers, limit cycles often do exist for systems
governed by a LME.  This is easily seen when we consider the special
case of Hamiltonian systems. In this case any eigenstate of the
Hamiltonian is a steady state but no other dynamical flows converge
to these steady states.  For the Bloch equation~(\ref{eq:bloch})
$\EEinv$ can be characterized explicitly. Consider the Jordan
decomposition of the Bloch superoperator, $\A=\sop{SJS}^{-1}$, where
$\sop{J}$ is the canonical Jordan form.  Let
$\gamma_\ell=\alpha_\ell+ i\beta_\ell$ be the eigenvalues of $\A$
and $\Pi_{\gamma}$ be the projector onto the (generalized)
eigenspace of the eigenvalue $\gamma$, and let $\I$ be the set of
indices of the eigenvalues of $\A$ with $\alpha_\ell=0$.

\begin{definition}
Let $\EEinv^{\lin}$ be the affine subspace of $\RR^{N^2-1}$
consisting of vectors of the form $\{\vec{s}_0+\vec{w}\}$, where
$\vec{s}_0$ is a solution of $\A\vec{s}_0+\vec{c}=\vec{0}$, and
$\vec{w}\in\EEcc$, where
$\EEcc=\sum_{\ell\in\I}\Pi_{\gamma_\ell}(\RR^{N^2-1})$ is the direct
sum of the eigenspaces of $\A$ corresponding to eigenvalues with
zero real part.  Then the invariant set $\EEinv=\EEinv^{\lin} \cap
\DD_{\RR}(\H)$.
\end{definition}

It is important to distinguish the invariant set $\EEinv$, which is
a set of Bloch vectors (or density operators), from the notion of an
invariant subspace of the Hilbert space $\H$.  In particular, as
$\EEinv$ contains the set of steady states $\EEss$, it is always
nonempty. Although $\supp(\EEinv)$, i.e., the union of the supports
of all states in $\EEinv$, is clearly an invariant subspace of $\H$,
in most cases $\supp(\EEinv)$ will be the entire Hilbert space.  In
particular, this is the case if $\EEinv$ contains a single state in
the interior, and $\supp(\EEinv)$ will be a proper subspace of the
Hilbert space only if all steady states are contained in a face at
the boundary.  This shows that proper invariant subspaces of the
Hilbert space exist only for systems that have steady states at the
boundary, and the maximal invariant subspace of the Hilbert space
can only be less than the entire Hilbert space if there are no
steady states in the interior.

\begin{theorem}
\label{thm:conv} Every trajectory $\vec{s}(t)$ of a system governed
by a Lindblad equation asymptotically converges to $\EEinv$.
\end{theorem}

\begin{proof}
Let $\vec{s}_0$ be a solution of the affine-linear equation
$\A\vec{s}_0 +\vec{c}=\vec{0}$, which exists by
Prop.~\ref{proposition:existence}.
$\vec{\Delta}(t)=\vec{s}(t)-\vec{s}_0$ satisfies the homogeneous
linear equation
$\dot{\vec{\Delta}}(t)=\A\vec{\Delta}(t)=\sop{SJS}^{-1}
\vec{\Delta}(t)$, where $\sop{J}=\diag(J_{\ell})$ is the Jordan
normal form of $\A$ consisting of irreducible Jordan blocks
$J_{\ell}$ of dimension $k_\ell$ with eigenvalue $\gamma_\ell$.
Setting $\vec{x}(t) =\sop{S}^{-1}\vec{\Delta}(t)$ gives
$\dot{\vec{x}}(t)=\sop{J}\vec{x}(t)$ and
$\vec{x}(t)=e^{t\sop{J}}\vec{x}(0)$, where $e^{t\sop{J}}$ is
block-diagonal with blocks
\begin{equation}
  \label{eq:Eell}
  E_{\ell}(t) = e^{t\alpha_\ell}
  \begin{bmatrix} R_\ell & t R_\ell & \frac{1}{2} t^2 R_\ell
                    & \frac{1}{6} t^3 R_\ell &\ldots \\
                  0 & R_\ell & tR_\ell & \frac{1}{2} t^2R_\ell &\ldots\\
                  0 & 0 & R_\ell & tR_\ell & \ldots \\
                  \vdots & \vdots & \vdots & \ddots
  \end{bmatrix},
\end{equation}
where $R_\ell=1$ if $\beta_\ell=0$, otherwise
\begin{equation}
  \label{eq:Rell}
  R_\ell = \begin{bmatrix}
            \cos(t\beta_\ell) & -\sin(t\beta_\ell) \\
            \sin(t\beta_\ell) & \cos(t\beta_\ell)
           \end{bmatrix}.
\end{equation}
Since the dynamical evolution is restricted to a bounded set, the
matrix $\A$ cannot have eigenvalues with positive real parts, i.e.,
$\alpha_\ell\le 0$, and taking the limit for $t\to\infty$ shows that
the Jordan blocks with $\alpha_\ell<0$ are annihilated, and thus
$\vec{\Delta}(t)\to \sop{S}\vec{x}_\infty=\vec{w}\in\EEcc$ and
$\vec{s}(t)=\vec{s}_0+\sop{S}\vec{x}(t)\to \vec{s}_0+\vec{w}$.
\end{proof}

The dimension of the invariant set, or more precisely, the affine
hyperplane of $\RR^{N^2-1}$ it belongs to, is equal to the sum of
the geometric multiplicities of the eigenvalues $\gamma_\ell$ with
zero real part, while the dimension of the set of steady states is
equal to the number of zero eigenvalues of $\A$.  Thus, in general,
the invariant set is much larger than the set of steady states of
the system.  Convergence to a steady state is guaranteed only if
$\A$ has no purely imaginary eigenvalues.  In particular, if all
eigenvalues of $\A$ have negative real parts, i.e., $\I=\emptyset$,
then all trajectories $\vec{s}(t)$ converge to the unique steady
state $\vec{s}_{\ss}= -\A^{-1}\vec{c}$. If $\A$ has purely imaginary
eigenvalues, then the steady states are centers and the invariant
set contains center manifolds, which exponentially attract the
dynamics~\cite{94Glendinning}.  In either case the trajectories of
the system are
\begin{equation}
  \vec{s}(t)=\vec{s}_0+\sop{S}e^{t\sop{J}}\sop{S}^{-1}(\vec{s}(0)-\vec{s}_0),
\end{equation}
and the distance of $\vec{s}(t)$ from the invariant subspace
\begin{equation}
  d(\vec{s}(t),\EEinv)=\norm{\sop{S}\vec{x}^\perp(t)},
\end{equation}
where $\vec{x}^\perp(t)=\sum_{\ell\not\in\I} E_{\ell}(t)
\Pi_{\gamma_\ell}(\vec{x}(0))$.  Equation.~(\ref{eq:Eell}) also
shows that any eigenvalue with zero real part cannot have a
nontrivial Jordan block as the dynamics would become unbounded
otherwise.  Thus the geometric and algebraic multiplicities of
eigenvalues with zero real part must agree. Moreover, the
eigenvalues of the (real) matrix $\A$ occur in complex conjugate
pairs $\gamma=\alpha\pm i\beta$. Thus, if $\A$ has a pair of
eigenvalues $\pm i\alpha$ with multiplicity $k$, then the center
manifold (as a subset of $\RR^{N^2-1}$) is at least $2k$
dimensional. Finally, as a unique steady state cannot be a center,
it follows that if $\A$ has purely imaginary eigenvalues, then it
must also have at least one zero eigenvalue, and there will be a
manifold of steady states, all of which are centers.  The properties
of the invariant set are nicely illustrated by the following
example.

Consider a four-level system with
$\dot\rho(t)=-i[H,\rho]+\D[V]\rho$, where
\begin{equation*}
 H = \frac{\sqrt{5}}{15}
     \begin{bmatrix}
      6&2&1&-2\\
      2&-6&2&1\\
      1&2&2&6\\
      -2&1&6&-2
      \end{bmatrix},  \quad
 V = \begin{bmatrix}
      1&-2&-1&1\\
      1&-1&-1&0\\
      0&-1&0&1\\
      1&-1&-1&0
      \end{bmatrix}.
\end{equation*}
$V$ is indecomposable and has two proper and two generalized
eigenvectors with eigenvalue $0$.  Let $\H_0$ be the subspace of
$\H$ spanned by the proper eigenvectors.  $H$ is blockdiagonal with
respect to a suitable orthonormal basis of $\H_0\oplus\H_0^\perp$,
and there is a 1D manifold of steady states
\begin{equation*}
 \rho_c(a) = \frac{1}{10}
 \begin{bmatrix}
  3-20\,a &-5\,a+1&2-15\,a&-5\,a+1\\
  -5\,a+1&10\,a+2&-15\,a-1&10\,a+2\\
  2-15\,a&-15\,a-1&3&-15\,a-1\\
 -5\,a+1&10\,a+2&-15\,a-1&10\,a+2
 \end{bmatrix}
\end{equation*}
where $a \in \frac{\sqrt{5}}{15}[-1,1]$.  We can verify that the
Bloch matrix $\A$ has a pair of purely imaginary eigenvalues $\pm
2i$ in addition to a $0$ eigenvalue and that the invariant set
$\EEinv$ consists of all density matrices with support on the
subspace $\H_0$ spanned by the proper eigenvectors of $V$ defined
above.  In terms of the corresponding Bloch vectors the invariant
set corresponds to the intersection of a three-dimensional invariant
subspace of $\RR^{15}$ with $\DD_\RR(\H)$.  This subspace is what we
refer to as the ``face'' at the boundary, although note that this
face is in fact homeomorphic to the 3D Bloch ball in this case.
Fig.~\ref{fig:conv}(a) shows that all trajectories converge to
$\EEinv$, but (b) shows that the trajectories do not converge to
steady states (except for a set of measure zero). Rather, states
starting outside the invariant set converge to paths in $\EEinv$,
which in this example are circular closed loops.  It is also
important to note that most initial states, even initially pure
states, converge to mixed states (with lower purity) with support on
the invariant set [see Fig.~\ref{fig:conv}(c)].

\begin{figure}
\begin{center}
\vspace{-0.1in}
\includegraphics[width=0.95\columnwidth]{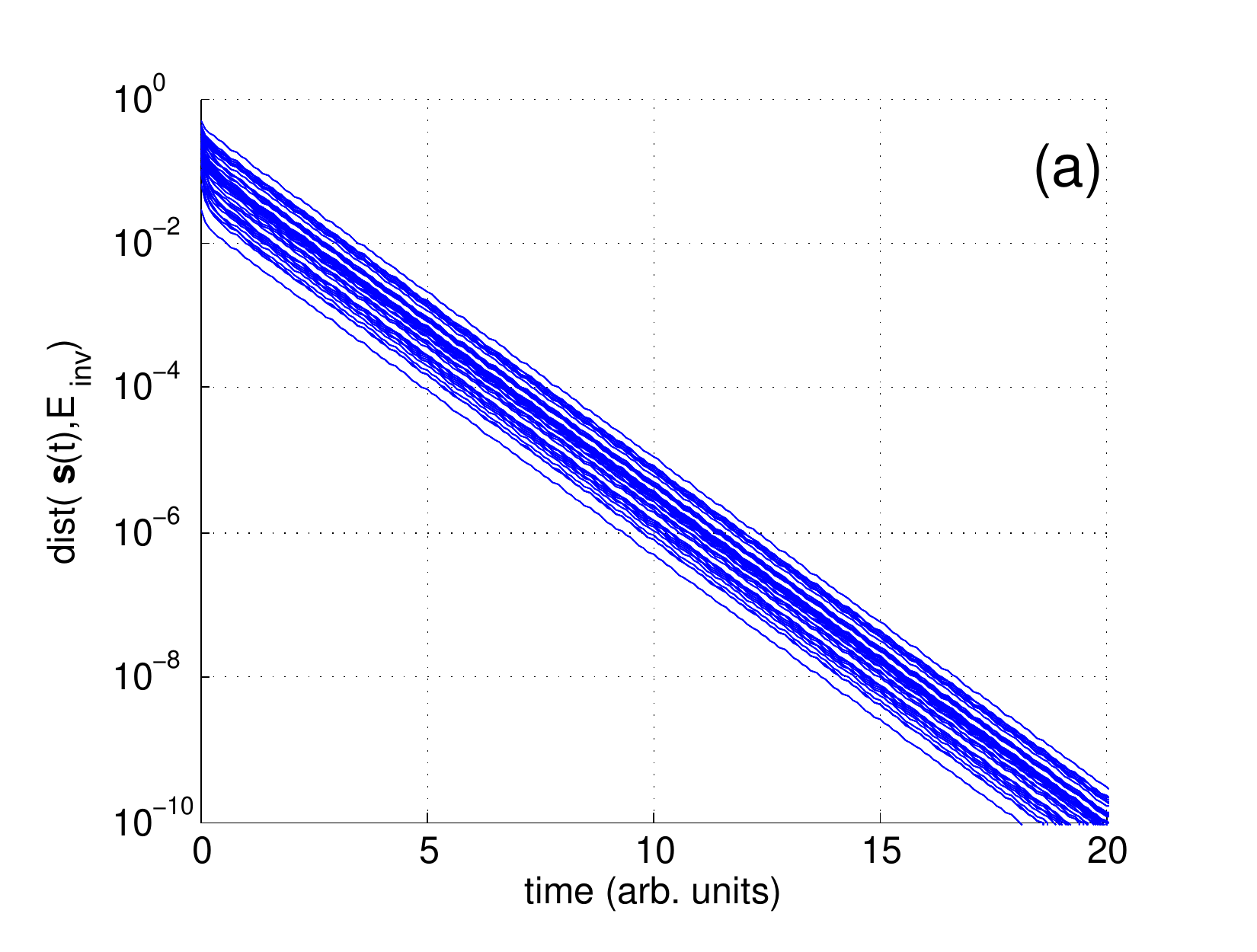}\\
\vspace{-0.3in}
\includegraphics[width=0.95\columnwidth]{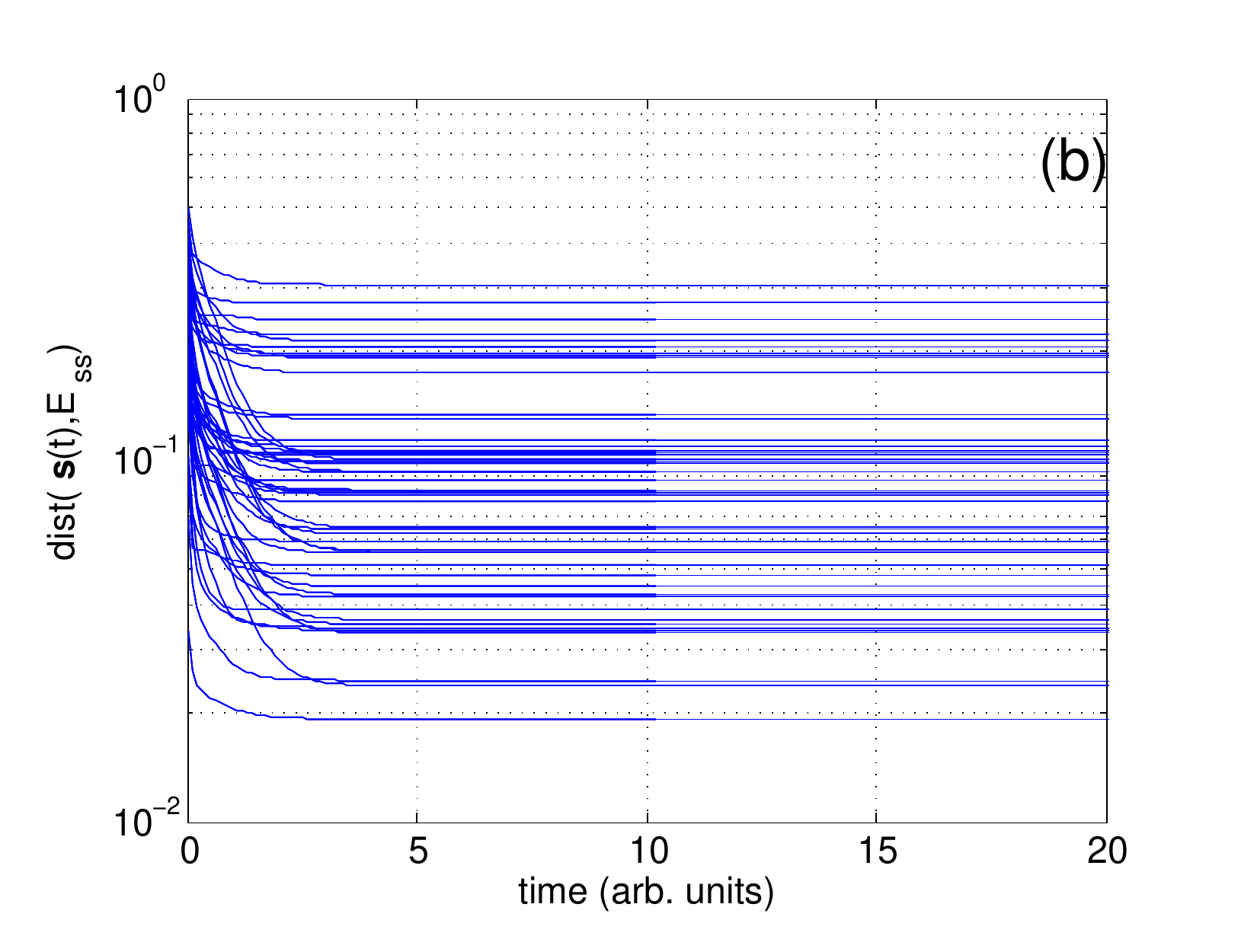}\\
\vspace{-0.3in}
\includegraphics[width=0.95\columnwidth]{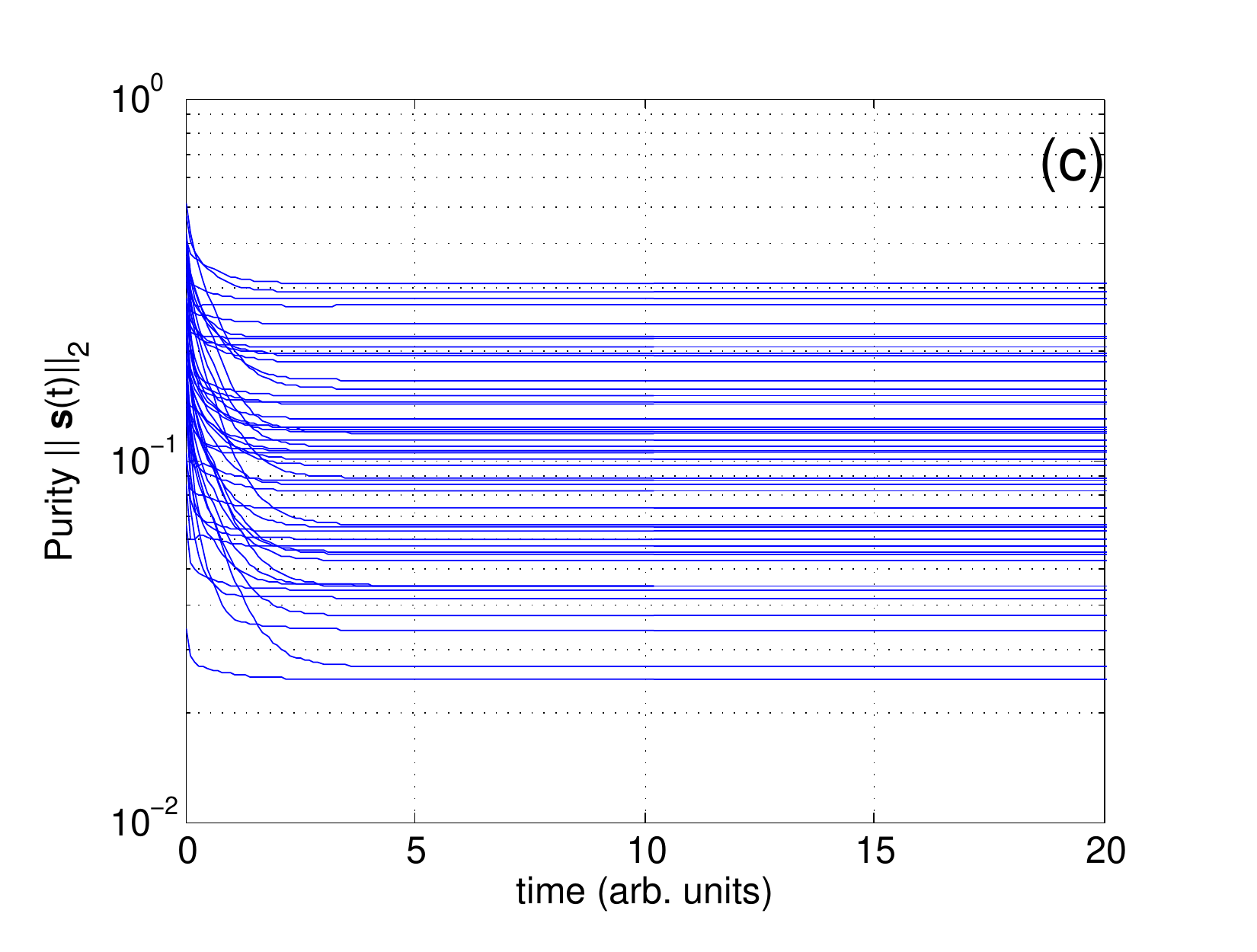}\\
\vspace{-0.1in}
\end{center}
\caption{(Color online) Semilogarithmic plots of (a) the distance of
$\vec{s}(t)$ from the invariant set $\EEinv$, (b) the distance from
the set of steady states, and (c) the purity $\norm{\vec{s}(t)}$ as
a function of time for 50 trajectories starting with 50 random
initial states $\vec{s}_0$.  (a) All of the trajectories converge to
the invariant set at a constant rate, indicating exponential decay
to the invariant set, but the distances of the trajectories from the
smaller set of steady states $\EEss \subset\EEinv$ in (b) do not
decrease to zero; rather they converge to different limiting values,
consistent with convergence of each trajectory to a different limit
cycle inside the invariant set.  As expected considering that the
set of steady states $\EEss$ is a measure-zero subset of $\EEinv$,
the limiting values of the distances are strictly positive, i.e.,
none of the 50 trajectories converges to a steady state.  (c) The
trajectories converge to various mixed states.  All limiting values
are far below $\frac{1}{2}\sqrt{3}$, the limiting value for a pure
state, i.e., none of the 50 trajectories converges to a pure state,
again as expected, as the set of pure states in $\EEinv$ is a
measure zero subset.}  \label{fig:conv}
\end{figure}

Although this example may seem rather artificial the properties of
the invariant set and the convergence behavior illustrated here are
relevant for real physical systems.  One important class of physical
systems with nontrivial invariant sets are those that possess
(non-trivial) decoherence-free subspaces (DFS).  By nontrivial we
mean here that $\EEinv$ or the DFS consists of more than one point.
A DFS $\H_{\DFS}$ is generally defined to be a subspace of the
Hilbert space $\H$ that is invariant under the dynamics and on which
we have unitary evolution.  In general this means that
$\L_D(\rho)=0$ if $\supp(\rho)\subset\H_{\DFS}$. Thus there should
exist a Hamiltonian $H$ and Lindblad operators $V_k$ such that
$H\ket{\psi}\in\H_{\DFS}$ for any $\ket{\psi}\in\H_{\DFS}$ and
$\sum_k \D[V_k]\rho=0$ for any $\rho$ with
$\supp(\rho)\subset\H_{\DFS}$. We must be careful, however, because
the decomposition of $\L_D$ is not unique and the Lindblad terms can
contribute to the Hamiltonian as we have already seen above, so the
effective Hamiltonian on the subspace may not be the same as the
system Hamiltonian without the bath.

\begin{proposition}
\label{prop:DFS} If a system governed by a LME has a DFS then
$\H_{\DFS} \subset\supp(\EEinv)$ and any state $\rho$ with
$\supp(\rho)\subset \H_{\DFS}$ belongs to $\EEinv$.
\end{proposition}

\begin{proof}
If $\H_{\DFS}$ is a proper subspace of $\H$ then the states with
support on it correspond to a face $F$ at the boundary of the state
space of Bloch vectors or positive unit-trace operators $\rho$.  As
$\H_{\DFS}$ is an invariant subspace of $\H$, the face $F$ must be
invariant under the dynamics, i.e., $\A\vec{s}+\vec{c}\in F$ for any
$\vec{s}\in F$, and thus the face $F$ must contain a steady state
$\vec{s}_{\ss}$ with $\A\vec{s}_{\ss}+\vec{c}=\vec{0}$.  Moreover,
there exists a subspace $S$ of $\RR^{N^2-1}$ such that for any
$\vec{s}\in F$ we have $\vec{s}=\vec{s}_{\ss}+\vec{v}$ with
$\vec{v}\in S$.  Let $\A_H$ and $\A_D$ be the Bloch operators
associated with the Hamiltonian $H$ and dissipative dynamics.  We
can take $\A_H$ and $\A_D$ to be the anti-symmetric and symmetric
parts of $\A$, respectively.  If $\rho$ is a state with support on
$\H_{\DFS}$ then its Bloch vector $\vec{s}$ must satisfy
\begin{align*}
 \A \vec{s} + \vec{c}
 &= \A \vec{s}_{\ss} + \A\vec{v} + \vec{c} = \A \vec{v}\\
 &= (\A_H + \A_D) \vec{v} = \A_H \vec{v}
\end{align*}
for all $\vec{v}\in S$.  Due to the invariance property we have
$\A_H\vec{v}\in S$ and as $\A_H$ is a real antisymmetric matrix, it
has purely imaginary eigenvalues.  This shows that $\vec{v}$ must be
a linear combination of eigenvectors of $\A$ with purely imaginary
eigenvalues, i.e., $\vec{v}\in\EEcc$ and $\vec{s}\in\EEinv$.
\end{proof}

There are many examples of systems that have decoherence-free
subspaces. For instance, in the example above we can verify that
$\H_0$ is a DFS as $\H_0$ is invariant under the Hamiltonian
dynamics and for any $\rho$ with support on $\H_0$ we have trivially
$\D[V]\rho=0$ as $\H_0$ is the subspace of $\H$ spanned by the two
(nonorthogonal) eigenvectors of $V$ with eigenvalue $0$.  Hence,
$\rho=\sum_{k=1,2}w_k\ket{\psi_k}\bra{\psi_k}$ and $V\ket{\psi_k}=0$
for $k=1,2$ implies
\begin{equation*}
   V\ket{\psi_k}\bra{\psi_k}V^\dag
 = V^\dag V \ket{\psi_k}\bra{\psi_k}
 = \ket{\psi_k}\bra{\psi_k}V^\dag V =0
\end{equation*}
and thus $\D(V)\rho=0$.

A simpler, more physical example is a three-level $\Lambda$ system
with decay of the excited state $\ket{2}$ given by the LME with
$H=\diag(0,1,0)$ and $V_1 = \ket{1}\bra{2}$, $V_2=\ket{3}\bra{2}$.
The system has a DFS spanned by the stable ground states
$\H_{\DFS}=\mbox{\rm span}\{\ket{1},\ket{3}\}$ as we clearly have
$V_1\ket{1}=V_1\ket{3}=0$ and $V_2\ket{1}=V_2\ket{3}=0$ and thus
$V_1\ket{\psi}=V_2\ket{\psi}=0$ for all
$\psi=\alpha\ket{1}+\beta\ket{3}$ and thus
$\D[V_1]\rho=\D[V_2]\rho=0$ for any
\begin{equation}
  \label{eq:rho_DFS}
  \rho=w_1 \ket{\psi_1}\bra{\psi_1}+w_2 \ket{\psi_2}\bra{\psi_2}, \quad
  \ket{\psi_k}\in \H_{\DFS},
\end{equation}
for $k=1,2$, and $\H_{\DFS}$ is invariant under the Hamiltonian $H$.
In this case it is easy to check that the corresponding invariant
set $\EEinv$ is precisely the face $F$ at the boundary corresponding
to density operators of the form~(\ref{eq:rho_DFS}).  In fact, as
the Hamiltonian is trivial on $\H_{\DFS}$, all of the states with
support on $\H_{\DFS}$ are actually steady states, i.e.,
$\EEinv=\EEss$.  This would no longer be the case if we changed the
Hamiltonian to $H'=\ket{1}\bra{3}+\ket{3}\bra{1}$, for instance, but
$\EEinv$ would still be an invariant set.  The requirement that
$\H_{\DFS}$ be invariant under the Hamiltonian dynamics is very
important.  If we change the Hamiltonian above to
$H''=\ket{1}\bra{2}+\ket{2}\bra{1}$, for example, then the system no
longer has a DFS.  In fact, it is easy to check that the invariant
set collapses to a single point, here
$\EEinv=\EEss=\{\ket{3}\bra{3}\}$.  Other choices of the Hamiltonian
will result in different steady states.  As the states with support
on a DFS must be contained in the invariant set $\EEinv$, only
systems with non-trivial $\EEinv$ admit DFS's.

\begin{figure}
\includegraphics[width=\columnwidth]{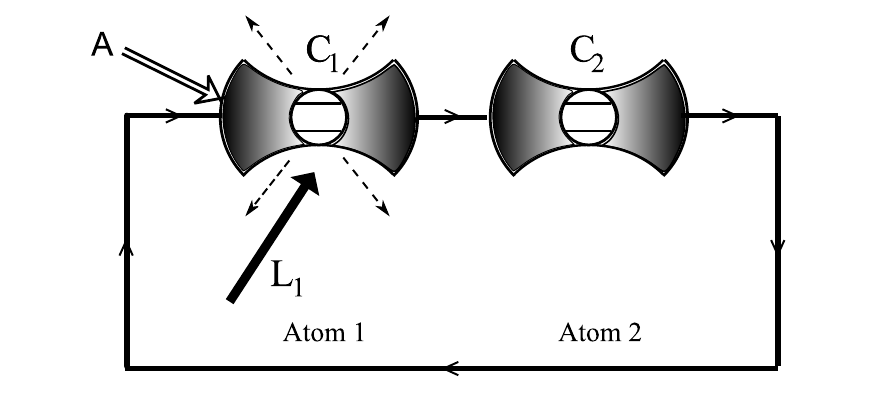}
\caption{Two atoms in separated cavities connected into a closed
loop through optical fibers. The off-resonant driving field $A$
generates an effective Hamiltonian $H_{eff}=Z_1Z_2$. Atom 1 is also
driven by a resonant laser field generating a local Hamiltonian
$X_1$. In the time scale we are interested in, only atom 1
experiences spontaneous decay. } \label{fig:one_decay_atom}
\end{figure}

Another example are two spins subject to the LME
\begin{align*}
  \dot \rho = -i\alpha[Z_1 Z_2,\rho]+\gamma_1\D[\sigma_1]\rho.
\end{align*}
where $\sigma_k$ is the decay operator for spin $k$ and $Z_k =
\sigma_k\sigma_k^\dag-\sigma_k^\dag \sigma_k$.  Here we have an
effective Ising interaction term and a decay term for the first
spin. This model might describe an electron spin weakly coupled to
stable a nuclear spin.  The same model was derived for two atoms in
separate cavities connected by optical fibers
(Fig~\ref{fig:one_decay_atom}) in the large-detuning regime
\cite{PhysRevA.70.022307,PhysRevA.80.042305}. In the latter case we
could achieve $\gamma_2\ll \gamma_1$ by choosing different $Q$
factors for the two cavities so that on certain time scale that one
atom experiences spontaneous decay while the other does not. For a
system of this type the Hilbert space has a natural tensor product
structure $\H=\H_1\otimes\H_2$ and we immediately expect the
invariant set to be $\{\rho=\ket{0}_1\bra{0}_1\otimes\rho_2\}$ as
subsystem 2 is clearly unaffected by dissipation,
$D[\sigma_1](\rho_1 \otimes \rho_2)=(\D[\sigma_1] \rho_1) \otimes
\rho_2$.  It is easy to see that $H$ is invariant on $\EEinv$, and
the states $\ket{0}\otimes \ket{\psi}$ with $\ket{\psi}\in \H_2$
also form a DFS.  However, if atom 1 is driven by a resonant laser
field $\Omega$ then the Lindblad dynamics becomes
\begin{align*}
\dot \rho = -i\alpha[ Z_1
Z_2,\rho]-i\Omega[X_1,\rho]+\gamma_1\D[\sigma_1]\rho
\end{align*}
and the DFS disappears.   The system still has a 1D manifold of
steady states but the Bloch superoperator $\A$ no longer has purely
imaginary eigenvalues $\pm i\gamma$ with $\gamma>0$, i.e., the
invariant set collapses to the 1D manifold of steady states.

\section{Conclusion}

We have theoretically investigated the convex set of the steady
states and the invariant set of the Lindblad master equation,
derived several sufficient conditions for the existence of a unique
steady state, and applied these to different physical systems.  One
interesting result is that if one Lindblad term corresponds to an
annihilation operator of the system then the stationary state is
unique.  Another useful result is that a composite system has a
unique steady state if the Lindblad equation contains dissipation
terms corresponding to annihilation operators for each subsystem. In
both cases the result still holds if other dissipation terms are
present, and regardless of the Hamiltonian. We also show that
uniqueness implies asymptotic stability of the steady state and
hence global attractivity.  On the other hand, if there are at least
two steady states, then there is a convex set of steady states, none
of which are asymptotically stable.  Furthermore, in this case even
convergence to a steady state is \emph{not} guaranteed as there can
be a larger invariant set surrounding the steady states,
corresponding to the center manifold generated by the eigenspaces of
the Bloch superoperator $\A$ with purely imaginary eigenvalues.  The
invariant set is closely related to decoherence-free subspaces; in
particular any state $\rho$ with support on a DFS belongs to the
invariant set.

This characterization of the set of steady states and the invariant
set, can be used to stabilize desired states using Hamiltonian and
reservoir engineering, and we illustate how in principle any state,
pure or mixed, can be stabilized this way.  This can be extended to
engineering decoherence-free subspaces.  The latter are naturally
attractive but attractivity of a subspace is a weak property in that
almost all initial states will generally converge to mixed state
trajectories with support on the subspace, not stationary pure
states.  One possibility of implementing such reservoir engineering
is via direct feedback, e.g., by homodyne detection, which yields a
feedback-modified master equation~\cite{pra49p2133} with Lindblad
terms depending on the measurement and feedback Hamiltonians.  This
dependence shows that feedback can change the reservoir operators,
and we have shown that in the absence of restrictions on the
control, measurement, and feedback operators, any state can be
rendered asympotically stable by means of direct feedback.  It will
be interesting to consider what states can be stabilized, for
example, given a restricted set of available measurement, control,
and feedback Hamiltonians for specific physical models.

\acknowledgments

S.G.S. acknowledges funding from EPSRC ARF Grant EP/D07192X/1, the
EPSRC QIP Interdisciplinary Research Collaboration (IRC), Hitachi,
and NSF Grant PHY05-51164. We gratefully acknowledge Francesco
Ticozzi, Howard Wiseman, Heide Narnhofer, Bernhard Baumgartner, Marj
Batchelor, and Tsung-Lung Tsai for various helpful discussions and
comments.

\appendix
\section{Existence of Steady States}
\label{app:ss1}

We can use Brouwer's fixed point theorem and Cantor's intersection
theorem to prove that any dynamical system whose flow is a
continuous map $\phi_t$ from a disk $D^n$ to itself, must have a
fixed point, and the assumption that the domain is the disk $D^n$
can be relaxed to any simple-connected compact set.  Specifically,
we have:

\begin{theorem}
Let $\dot x=f(x)$ be a dynamical system with a flow $\phi_t$ from a
simply connected compact set $D$ to itself.  If $\phi_t$ is
continuous, then there exists a fixed point.
\end{theorem}

\begin{proof}
For any given $T>0$, $\phi_T:D\to D$ is a continuous map from $D$ to
itself.  Applying Brouwer's fixed point theorem, there exists at
least one fixed point.  Denote the set of fixed points as $S_T$ and
observe that as a closed subset of a compact set $S_T$ is compact.
Similarly, we can find the set of fixed points $S_{T/2}$ for
$\phi_{T/2}$, which is also compact and satisfies
$S_{\frac{t}{2}}\subset S_T$ as a fixed point of $\phi_{T/2}$ is
also a fixed point of $\phi_T$.  By iterating this procedure we can
construct a sequence of nonempty compact netting sets
$\{S_{T/2^k}\}:\cdots \subset S_{T/2^k}\subset S_{T/2^{k-1}} \subset
\cdots S_{\frac{T}{2}}\subset S_T$.  By Cantor intersection theorem,
the intersection of $\{S_{T/2^k}\}$ is nonempty.  Let $x_0$ be one
of the points in the intersection.  Then for any $T'=nT/2^k$, we
have $\phi_{T'}(x_0)=x_0$. Since such $T'$ is dense for
$[0,+\infty)$ and $\phi_t$ a continuous flow, we know that for any
time $t$, $\phi_t(x_0)=x_0$, i.e. $x_0$ is a fixed of the dynamical
system.
\end{proof}

Since the set of physical states is a compact simply connected set
and the master equation clearly continuous, we can conclude that any
system governed by a Lindblad master equation has a physical
stationary state.

\section{Extremal points of Convex Set of Steady States}
\label{app:extremals}

\begin{lemma}
\label{lemma:app1} If $\rho_s = s\rho_0+(1-s)\rho_1$ is a convex
combination of the positive operators $\rho_0$, $\rho_1$ with
$0<s<1$ then $\rank\rho_s$ is constant and the support of $\rho_0$
and $\rho_1$ is contained in the support of $\rho_s$.
\end{lemma}

\begin{proof}
If $\rank(\rho_s)=k$ then there exists a basis such that
$\rho_s=\diag(r_1,\ldots,r_k,0,\ldots)$ with $r_\ell\ge 0$ and
$\sum_{\ell=1}^k r_k=1$, i.e., the last $N-k$ rows and columns of
$\rho_s$ are $0$.  Since $\rho_0$ and $\rho_1$ are positive
operators and $s>0$, this is possible only if the last $N-k$ rows
and columns of $\rho_0$ and $\rho_1$ are zero, and thus the support
of $\rho_0$ and $\rho_1$ is contained in the support of $\rho_s$.
Furthermore, the rank of all $\rho_s$ on the open line segment
$0<s<1$ must be the same.  If there were two intermediate points
with $\rank(\rho_s)< \rank(\rho_t)$ and $0<s<t<1$ then the support
of $\rho_0$ and $\rho_t$ would have to be contained in the support
of $\rho_s$ by the previous argument, which is impossible as
$\rank(\rho_s)<\rank(\rho_t)$. Similarly, for $0<t<s<1$.
\end{proof}

\begin{theorem}
\label{thm:Hs_decomp} Let $\H_s$ be the smallest subspace of $\H$
that contains the support of all steady states.  There exist a
finite number of extremal steady states $\rho_k$ such that $\EEss$
is the convex hull of $\{\rho_k\}$ and $\H_s=\oplus_k
\supp(\rho_k)$.
\end{theorem}

\begin{proof}
We know that a convex set is the convex hull of its extremal points
but there may be many extremal points with nonorthogonal supports.
Thus, what we need to show is that we can always choose a subset of
the extremal points with mutually orthogonal supports that generates
the entire convex set of steady states.  Given two extremal steady
states $\rho_1$, $\rho_2$, either $\supp(\rho_1)\perp\supp(\rho_2)$,
or we can find another steady state $\rho_3$ with
$\supp(\rho_2)\subset \supp(\rho_1)+\supp(\rho_2)$ such that
$\supp(\rho_1) \perp \supp(\rho_3)$.  Assuming we have already
constructed $\H_0= \oplus_\ell^{k-1} \supp(\rho_\ell)$ with
$\supp(\rho_\ell)$ mutually orthogonal, let $\rho_k$ be another
extremal point with $\supp(\rho_k)$ not included in $\H_0$ and
define $\H_1=\H_0+\supp(\rho_k)$.  By connecting $\rho_k$ and a
fixed point with full rank in $\H_0$, we can find another steady
state with full rank in $\H_1$.  So $\H_1$ is an invariant subspace
under the dynamics on $\H$.  Therefore, in the following, we will
restrict the dynamics $H$ and $V_k$ on the subspace $\H_1$.  Define
$P$ to be the projection operator of $\H_0$ and $P^\perp$ to be the
orthogonal projection operator of $P$ with respect to $\H_1$. In the
block-diagonal diagonal form with respect to $P$ and $P^\perp$,

\begin{align*}
\rho_k=\begin{bmatrix}
\rho_{11} & \rho_{12} \\
\rho_{21} & \rho_{22}
\end{bmatrix},
H=\begin{bmatrix}
H_{11} & H_{12} \\
H_{21} & H_{22}
\end{bmatrix},
V=\begin{bmatrix}
V_{11} & V_{12} \\
V_{21} & V_{22}
\end{bmatrix},
\end{align*}
where without loss of generality we only consider one Lindblad term.
Since $\H_0$ is an invariant subspace under the dynamics, we have
\begin{subequations}
\label{fixed-condi}
\begin{align}
0 &= V_{21}\\
0 &= -\frac{1}{2} \sum_j V_{11}^\dag V_{12}+i H_{12}
\end{align}
\end{subequations}

For a steady state $\rho_k$, we have
$0=\dot\rho=-i[H,\rho_k]+\D[V]\rho$.  Since $\rho_k$ is an extremal
point, $\rho_{22}$ has full rank.  Moreover, as $\rho_k$ is
stationary in $\H_1$, it is also stationary restricted to a subspace
$P^\perp\H_1$, which means
\begin{align*}
0=-i[H_{22},\rho_{22}]+\D[V_{22}]\rho_{22}.
\end{align*}
Substituting this as well as (\ref{fixed-condi}) into
$\dot\rho_{22}=0$, we find
\begin{align*}
\rho_{22} V_{12}^\dagger V_{12}+V_{12}^\dagger V_{12}\rho_{22}=0
\end{align*}
which means $V_{12}=0$ since $\rho_{22}$ has full rank.  Together
with (\ref{fixed-condi}) we have $[H,P]=[V,P]=0$. Hence $P^\perp
\H_1$ is also an invariant space under the dynamics restricted on
$\H_1$. Combining the condition that $\H_1$ is an invariant subspace
under the dynamics on $\H$, we conclude that $P^\perp \H_1$ is also
an invariant subspace under the dynamics on $\H$.  There must exist
an extremal fixed point $\bar\rho_k$ in $P^\perp \H_1$ with support
orthogonal to $\H_0$.  We have
$\bar\H_1=H_0\oplus\supp(\bar\rho_k)$. Continuing this process until
all fixed points are included in $\oplus_k \supp(\rho_k)$, we
finally obtain $\H_s=\oplus_k \supp(\rho_k)$.  This construction can
be completed in a finite number of steps as the dimension of $\H_s$
is finite.
\end{proof}

\section{Proof of ``No Isolated Centers'' Theorem}
\label{app:no_iso_centers}

Suppose $\A$ has a pair of purely imaginary eigenvalues $\pm i
\alpha$. Let $\vec{e}$ be an eigenvector of $\A$ corresponding to
the eigenvalue $+i\alpha$ with $\alpha>0$, i.e.,
$\L_{tot}(\vec{e})=\A \vec{e}=i\alpha \vec{e}$. In the Schrodinger
picture, we have $e^{t\A}\vec{e} =e^{i\alpha t}\vec{e}$. Let $E$ be
the operator, corresponding to $\vec{e}$, in the adjoint operator
space, with $\L_{tot}^\dag(E)=i\alpha E$. In the Heisenberg picture,
the adjoint dynamics gives $E(t)=e^{i\alpha t}E$, and $E(t)^\dag
E(t)=E^\dag E$, with $E^\dag E$ always positive. We can scale $E$
such that $\norm{E^\dag E}_\infty=1$. Thus, $E^\dag E$ is a positive
matrix with maximum eigenvalue $\lambda_{\max}=1$ and $\ket{\phi_0}$
as the associated eigenvector. Hence, we have $E^\dag
E\ket{\phi_0}=\ket{\phi_0}$ and $\Tr(E^\dag
E\rho_0)=\Tr(\lambda_{\max}\rho_0)=1$, where
$\rho_0=\ket{\phi_0}\bra{\phi_0}$. Let us consider in the
Schrodinger picture, the evolution of $\rho(t)$ with initial state
$\rho(0)=\rho_0$. We define the average state
\begin{equation}
  \bar\rho(T) = \frac{1}{T} \int_0^T \rho(t) \, dt.
\end{equation}
Setting $D=E^\dag E$, switching between the Schrodinger and
Heisenberg picture, and using the Kadison inequality $D(t) \ge
E(t)^\dag E(t)$ we obtain
\begin{align*}
  \Tr(\bar\rho D)
  &= \frac{1}{T}\int_0^T \Tr[\rho(t) D] \, dt       \\
  &= \frac{1}{T}\int_0^T \Tr[\rho D(t)] \, dt       \\
  &\ge \frac{1}{T}\int_0^T \Tr[\rho_0 E(t)^\dag E(t)] \, dt \\
  &= \frac{1}{T}\int_0^T \Tr[\rho_0 E^\dag E] \, dt       \\
  &= \Tr(\rho_0 E^\dag E) = \norm{E^\dag E}_\infty = 1.
\end{align*}
On the other hand, we have $\Tr(\bar\rho E^\dag E)\le 1$, and thus
$\Tr(\bar\rho E^\dag E)=1$.

If the unique steady state $\rho_{\ss}$ is in the interior of the
convex set of physical states, i.e., $\rho_{\ss}$ has full rank,
then $\bar\rho$ must have full rank for sufficiently large $T$ as
well, and this is possible only if $E^\dag E=\ONE$, i.e., $E$ is
unitary. Next we calculate the term $E^\dag \L_{tot}^\dag(E)$. From
the evolution:
\begin{align*}
\L_{tot}^\dag (E)=[iH,E]+\sum_j \Big(V_j^\dag E
V_j-\frac{1}{2}(EV_j^\dag V_j+V_j^\dag V_jE)\Big),
\end{align*}
we have
\begin{equation}
 \label{eqn:lind_adjoint}
 \begin{split}
  E^\dag \L_{tot}^\dag (E) =& E^\dag [iH,E]\nonumber \\
  &+ \sum_j \Big(E^\dag V_j^\dag E V_j-\frac{1}{2}V_j^\dag
V_j-\frac{1}{2}E^\dag
 V_j^\dag V_jE\Big)
\end{split}
\end{equation}
Since $E$ is unitary, we have $E^\dag \L_{tot}^\dag (E)=E^\dag
i\alpha E=i \alpha I$. Taking the trace at both sides in
(\ref{eqn:lind_adjoint}),
\begin{align*}
\sum_j \Tr(E^\dag V_j^\dag E V_j)
 =&\frac{1}{2}\sum_j \left[\Tr(E^\dag EV_j^\dag V_j)
   +\Tr(E^\dag V_j^\dag V_jE)\right]\\
  &+i\alpha N
\end{align*}
On the other hand, by Cauchy Schwartz inequality,
\begin{align*}
&|\sum_j\Tr(E^\dagger V_j^\dagger E V_j)|\\
&\le \sum_j|\Tr\big((E^\dagger V_j^\dagger E) (V_j)\big)|\\
&\le \sum_j\sqrt{\Tr( E^\dagger V_j^\dagger EE^\dagger V_j
E)}\sqrt{\Tr( V_j^\dagger
V_j)}\\
&\le \sum_j \frac{1}{2}\Big[\Tr(E^\dagger V_j^\dagger V_j E)+\Tr(
V_j^\dagger V_j)\Big]
\end{align*}
Therefore, we must have $\alpha=0$, which contradicts the initial
assumption that $\alpha>0$.

When the unique steady state is a mixed state at the boundary with
$1<\rank(\rho_{\ss})<N$, then we can partition the Hilbert space
$\H=\H_1\oplus\H_2$ such that $\rho_{\ss}$ vanishes on $\H_2$.  It
has been shown that in this case all solutions are attracted to
states with support on $\H_1$.  Thus, we can restrict the dynamics
to $\H_1$, i.e., the support of $\rho_{\ss}$, and the same arguments
as above imply that $E^\dag E$ must equal the identity on the $\H_1$
subspace, $E^\dag E|_{\H_1}=\ONE_{\H_1}$, which leads to a
contradiction.

If the unique fixed point $\rho_{\ss}$ happens to be a pure state at
the boundary, then it is easy to see that there cannot be any loop
paths, because the state $\bar\rho$ averaged over one period would
have to equal the stationary state $\rho_{\ss}$, which is not
possible because a rank $1$ projector cannot be written as a linear
combination of other states.

Moreover, if $\vec{s}_{\ss}$ is a center that belongs to a face $F$
in the boundary, then the entire center manifold it belongs to must
be contained in $F$ as otherwise there would be loop planes
intersecting the boundary and physical states evolving into
non-physical states, which is forbidden.


\begin{thebibliography}{31}
\expandafter\ifx\csname natexlab\endcsname\relax\def\natexlab#1{#1}\fi
\expandafter\ifx\csname bibnamefont\endcsname\relax
  \def\bibnamefont#1{#1}\fi
\expandafter\ifx\csname bibfnamefont\endcsname\relax
  \def\bibfnamefont#1{#1}\fi
\expandafter\ifx\csname citenamefont\endcsname\relax
  \def\citenamefont#1{#1}\fi
\expandafter\ifx\csname url\endcsname\relax
  \def\url#1{\texttt{#1}}\fi
\expandafter\ifx\csname urlprefix\endcsname\relax\def\urlprefix{URL }\fi
\providecommand{\bibinfo}[2]{#2}
\providecommand{\eprint}[2][]{\url{#2}}

\bibitem[{\citenamefont{Schirmer}(2006)}]{qph0602014}
\bibinfo{author}{\bibfnamefont{S.~G.} \bibnamefont{Schirmer}}, in
  \emph{\bibinfo{booktitle}{Lagrangian and Hamiltonian Methods for Nonlinear
  Control 2006}}, edited by
  \bibinfo{editor}{\bibfnamefont{K.}~\bibnamefont{Fujimoto}} \bibnamefont{and}
  \bibinfo{editor}{\bibfnamefont{F.}~\bibnamefont{Bullo}},
  \bibinfo{organization}{IFAC} (\bibinfo{publisher}{Springer},
  \bibinfo{address}{Berlin/Heidelberg}, \bibinfo{year}{2006}), vol.
  \bibinfo{volume}{366} of \emph{\bibinfo{series}{Lecture Notes in Control and
  Information Sciences}}, p. \bibinfo{pages}{293},
  \urlprefix\url{http://arxiv.org/abs/quant-ph/0602014}.

\bibitem[{\citenamefont{Viola et~al.}(1999)\citenamefont{Viola, Knill, and
  Lloyd}}]{prl.82.2417}
\bibinfo{author}{\bibfnamefont{L.}~\bibnamefont{Viola}},
  \bibinfo{author}{\bibfnamefont{E.}~\bibnamefont{Knill}}, \bibnamefont{and}
  \bibinfo{author}{\bibfnamefont{S.}~\bibnamefont{Lloyd}},
  \bibinfo{journal}{Phys. Rev. Lett.} \textbf{\bibinfo{volume}{82}},
  \bibinfo{pages}{2417} (\bibinfo{year}{1999}).

\bibitem[{\citenamefont{Nigmatullin and Schirmer}(2009)}]{NJP11n105032}
\bibinfo{author}{\bibfnamefont{R.}~\bibnamefont{Nigmatullin}} \bibnamefont{and}
  \bibinfo{author}{\bibfnamefont{S.~G.} \bibnamefont{Schirmer}},
  \bibinfo{journal}{New J. Physics} \textbf{\bibinfo{volume}{11}},
  \bibinfo{pages}{105032} (\bibinfo{year}{2009}).

\bibitem[{\citenamefont{Gorini et~al.}(1976)\citenamefont{Gorini, Kossakowsi,
  and Sudarshan}}]{JMP.76821}
\bibinfo{author}{\bibfnamefont{V.}~\bibnamefont{Gorini}},
  \bibinfo{author}{\bibfnamefont{A.}~\bibnamefont{Kossakowsi}},
  \bibnamefont{and}
  \bibinfo{author}{\bibfnamefont{E.}~\bibnamefont{Sudarshan}},
  \bibinfo{journal}{J. Math. Phys.} \textbf{\bibinfo{volume}{17}},
  \bibinfo{pages}{821} (\bibinfo{year}{1976}).

\bibitem[{\citenamefont{Lindblad}(1976)}]{CommMathPhys.76119}
\bibinfo{author}{\bibfnamefont{G.}~\bibnamefont{Lindblad}},
  \bibinfo{journal}{Commun. Math. Phys.} \textbf{\bibinfo{volume}{48}},
  \bibinfo{pages}{119} (\bibinfo{year}{1976}).

\bibitem[{\citenamefont{Wang and Wiseman}(2001)}]{pra.64.063810}
\bibinfo{author}{\bibfnamefont{J.}~\bibnamefont{Wang}} \bibnamefont{and}
  \bibinfo{author}{\bibfnamefont{H.~M.} \bibnamefont{Wiseman}},
  \bibinfo{journal}{Phys. Rev. A} \textbf{\bibinfo{volume}{64}},
  \bibinfo{pages}{063810} (\bibinfo{year}{2001}).

\bibitem[{\citenamefont{Wang et~al.}(2005{\natexlab{a}})\citenamefont{Wang,
  Wiseman, and Milburn}}]{pra.71.042309}
\bibinfo{author}{\bibfnamefont{J.}~\bibnamefont{Wang}},
  \bibinfo{author}{\bibfnamefont{H.~M.} \bibnamefont{Wiseman}},
  \bibnamefont{and} \bibinfo{author}{\bibfnamefont{G.~J.}
  \bibnamefont{Milburn}}, \bibinfo{journal}{Phys. Rev. A}
  \textbf{\bibinfo{volume}{71}}, \bibinfo{pages}{042309}
  (\bibinfo{year}{2005}{\natexlab{a}}).

\bibitem[{\citenamefont{Wang and Wiseman}(2005)}]{EuropeanPJD.32.257}
\bibinfo{author}{\bibfnamefont{J.}~\bibnamefont{Wang}} \bibnamefont{and}
  \bibinfo{author}{\bibfnamefont{H.~M.} \bibnamefont{Wiseman}},
  \bibinfo{journal}{Euro. Phys. J. D} \textbf{\bibinfo{volume}{32}},
  \bibinfo{pages}{257} (\bibinfo{year}{2005}).

\bibitem[{\citenamefont{Carvalho and Hope}(2007)}]{pra.76.010301}
\bibinfo{author}{\bibfnamefont{A.~R.~R.} \bibnamefont{Carvalho}}
  \bibnamefont{and} \bibinfo{author}{\bibfnamefont{J.~J.} \bibnamefont{Hope}},
  \bibinfo{journal}{Phys, Rev, A} \textbf{\bibinfo{volume}{76}},
  \bibinfo{eid}{010301} (\bibinfo{year}{2007}).

\bibitem[{\citenamefont{Carvalho et~al.}(2008)\citenamefont{Carvalho, Reid, and
  Hope}}]{pra.78.012334}
\bibinfo{author}{\bibfnamefont{A.~R.~R.} \bibnamefont{Carvalho}},
  \bibinfo{author}{\bibfnamefont{A.~J.~S.} \bibnamefont{Reid}},
  \bibnamefont{and} \bibinfo{author}{\bibfnamefont{J.~J.} \bibnamefont{Hope}},
  \bibinfo{journal}{Phys. Rev, A} \textbf{\bibinfo{volume}{78}},
  \bibinfo{eid}{012334} (\bibinfo{year}{2008}).

\bibitem[{\citenamefont{Wiseman}(1994)}]{pra49p2133}
\bibinfo{author}{\bibfnamefont{H.~M.} \bibnamefont{Wiseman}},
  \bibinfo{journal}{Phys. Rev. A} \textbf{\bibinfo{volume}{49}},
  \bibinfo{pages}{2133} (\bibinfo{year}{1994}).

\bibitem[{\citenamefont{Misra and Sudarshan}(1977)}]{Sudarshan-zeno}
\bibinfo{author}{\bibfnamefont{B.}~\bibnamefont{Misra}} \bibnamefont{and}
  \bibinfo{author}{\bibfnamefont{E.~C.~G.} \bibnamefont{Sudarshan}},
  \bibinfo{journal}{J. Math. Phys.} \textbf{\bibinfo{volume}{18}},
  \bibinfo{pages}{756} (\bibinfo{year}{1977}).

\bibitem[{\citenamefont{Lidar and Whaley}(2003)}]{DFS}
\bibinfo{author}{\bibfnamefont{D.~A.} \bibnamefont{Lidar}} \bibnamefont{and}
  \bibinfo{author}{\bibfnamefont{B.~K.} \bibnamefont{Whaley}}, in
  \emph{\bibinfo{booktitle}{Irreversible Quantum Dynamics}}
  (\bibinfo{publisher}{Springer Berlin / Heidelberg}, \bibinfo{year}{2003}),
  vol. \bibinfo{volume}{622} of \emph{\bibinfo{series}{Lecture Notes in
  Physics}}, p. \bibinfo{pages}{293}.

\bibitem[{\citenamefont{Beige et~al.}(2000)\citenamefont{Beige, Braun, and
  Knight}}]{Almut-drive}
\bibinfo{author}{\bibfnamefont{A.}~\bibnamefont{Beige}},
  \bibinfo{author}{\bibfnamefont{D.}~\bibnamefont{Braun}}, \bibnamefont{and}
  \bibinfo{author}{\bibfnamefont{P.~L.} \bibnamefont{Knight}},
  \bibinfo{journal}{New J. Phys.} \textbf{\bibinfo{volume}{2}},
  \bibinfo{pages}{22} (\bibinfo{year}{2000}).

\bibitem[{\citenamefont{Kraus et~al.}(2008)\citenamefont{Kraus, B\"uchler,
  Diehl, Kantian, Micheli, and Zoller}}]{PhysRevA.78.042307}
\bibinfo{author}{\bibfnamefont{B.}~\bibnamefont{Kraus}},
  \bibinfo{author}{\bibfnamefont{H.~P.} \bibnamefont{B\"uchler}},
  \bibinfo{author}{\bibfnamefont{S.}~\bibnamefont{Diehl}},
  \bibinfo{author}{\bibfnamefont{A.}~\bibnamefont{Kantian}},
  \bibinfo{author}{\bibfnamefont{A.}~\bibnamefont{Micheli}}, \bibnamefont{and}
  \bibinfo{author}{\bibfnamefont{P.}~\bibnamefont{Zoller}},
  \bibinfo{journal}{Phys. Rev. A} \textbf{\bibinfo{volume}{78}},
  \bibinfo{pages}{042307} (\bibinfo{year}{2008}).

\bibitem[{\citenamefont{Spohn}(1976)}]{RepMathPhys.10.76189}
\bibinfo{author}{\bibfnamefont{H.}~\bibnamefont{Spohn}}, \bibinfo{journal}{Rep.
  Math. Phys.} \textbf{\bibinfo{volume}{10}}, \bibinfo{pages}{189}
  (\bibinfo{year}{1976}).

\bibitem[{\citenamefont{Spohn}(1977)}]{LettMathPhys.2.7733}
\bibinfo{author}{\bibfnamefont{H.}~\bibnamefont{Spohn}},
  \bibinfo{journal}{Lett. Math. Phys.} \textbf{\bibinfo{volume}{2}},
  \bibinfo{pages}{33} (\bibinfo{year}{1977}).

\bibitem[{\citenamefont{Frigerio}(1977)}]{LettMathPhys.2.7779}
\bibinfo{author}{\bibfnamefont{A.}~\bibnamefont{Frigerio}},
  \bibinfo{journal}{Lett. Math. Phys.} \textbf{\bibinfo{volume}{2}},
  \bibinfo{pages}{79} (\bibinfo{year}{1977}).

\bibitem[{\citenamefont{Frigerio}(1978)}]{CommMathPhys.63.78269}
\bibinfo{author}{\bibfnamefont{A.}~\bibnamefont{Frigerio}},
  \bibinfo{journal}{Commun. Math. Phys.} \textbf{\bibinfo{volume}{63}},
  \bibinfo{pages}{269} (\bibinfo{year}{1978}).

\bibitem[{\citenamefont{Davies}(1970)}]{CommMathPhys.19.7083}
\bibinfo{author}{\bibfnamefont{E.~B.} \bibnamefont{Davies}},
  \bibinfo{journal}{Commun. Math. Phys.} \textbf{\bibinfo{volume}{19}},
  \bibinfo{pages}{83} (\bibinfo{year}{1970}).

\bibitem[{\citenamefont{Evans}(1977)}]{CommMathPhys.54.77293}
\bibinfo{author}{\bibfnamefont{D.~E.} \bibnamefont{Evans}},
  \bibinfo{journal}{Commun. Math. Phys.} \textbf{\bibinfo{volume}{54}},
  \bibinfo{pages}{293} (\bibinfo{year}{1977}).

\bibitem[{\citenamefont{Baumgartner
  et~al.}(2008{\natexlab{a}})\citenamefont{Baumgartner, Narnhofer, and
  Thirring}}]{JPA.41.065201}
\bibinfo{author}{\bibfnamefont{B.}~\bibnamefont{Baumgartner}},
  \bibinfo{author}{\bibfnamefont{H.}~\bibnamefont{Narnhofer}},
  \bibnamefont{and} \bibinfo{author}{\bibfnamefont{W.}~\bibnamefont{Thirring}},
  \bibinfo{journal}{J. Phys A} \textbf{\bibinfo{volume}{41}},
  \bibinfo{pages}{065201} (\bibinfo{year}{2008}{\natexlab{a}}).

\bibitem[{\citenamefont{Baumgartner and Narnhofer}(2008)}]{JPA.41.395303}
\bibinfo{author}{\bibfnamefont{B.}~\bibnamefont{Baumgartner}} \bibnamefont{and}
  \bibinfo{author}{\bibfnamefont{H.}~\bibnamefont{Narnhofer}},
  \bibinfo{journal}{J. Phys. A} \textbf{\bibinfo{volume}{41}},
  \bibinfo{pages}{395303} (\bibinfo{year}{2008}).

\bibitem[{\citenamefont{Ticozzi and
  Viola}(2008{\natexlab{a}})}]{IEEE.53.082048}
\bibinfo{author}{\bibfnamefont{F.}~\bibnamefont{Ticozzi}} \bibnamefont{and}
  \bibinfo{author}{\bibfnamefont{L.}~\bibnamefont{Viola}},
  \bibinfo{journal}{IEEE Trans. on Autom. Control}
  \textbf{\bibinfo{volume}{53}}, \bibinfo{pages}{2048}
  (\bibinfo{year}{2008}{\natexlab{a}}).

\bibitem[{\citenamefont{Ticozzi and
  Viola}(2008{\natexlab{b}})}]{Automatica}
\bibinfo{author}{\bibfnamefont{F.}~\bibnamefont{Ticozzi}} \bibnamefont{and}
  \bibinfo{author}{\bibfnamefont{L.}~\bibnamefont{Viola}},
  \bibinfo{journal}{Automatica} \textbf{\bibinfo{volume}{45}}, \bibinfo{pages}{2002}
  (\bibinfo{year}{2009}{\natexlab{b}}).

\bibitem[{\citenamefont{Wang and
  Schirmer}(2009{\natexlab{a}})}]{PhysRevA.79.052326}
\bibinfo{author}{\bibfnamefont{X.}~\bibnamefont{Wang}} \bibnamefont{and}
  \bibinfo{author}{\bibfnamefont{S.~G.} \bibnamefont{Schirmer}},
  \bibinfo{journal}{Phys. Rev. A} \textbf{\bibinfo{volume}{79}},
  \bibinfo{pages}{052326} (\bibinfo{year}{2009}{\natexlab{a}}).

\bibitem[{\citenamefont{Wang et~al.}(2005{\natexlab{b}})\citenamefont{Wang,
  Wiseman, and Milburn}}]{pra71n042309}
\bibinfo{author}{\bibfnamefont{J.}~\bibnamefont{Wang}},
  \bibinfo{author}{\bibfnamefont{H.~M.} \bibnamefont{Wiseman}},
  \bibnamefont{and} \bibinfo{author}{\bibfnamefont{G.~J.}
  \bibnamefont{Milburn}}, \bibinfo{journal}{Phys. Rev. A}
  \textbf{\bibinfo{volume}{71}}, \bibinfo{pages}{042309}
  (\bibinfo{year}{2005}{\natexlab{b}}).

\bibitem[{\citenamefont{Baumgartner
  et~al.}(2008{\natexlab{b}})\citenamefont{Baumgartner, Narnhofer, and
  Thirring}}]{JPA41n065201}
\bibinfo{author}{\bibfnamefont{B.}~\bibnamefont{Baumgartner}},
  \bibinfo{author}{\bibfnamefont{H.}~\bibnamefont{Narnhofer}},
  \bibnamefont{and} \bibinfo{author}{\bibfnamefont{W.}~\bibnamefont{Thirring}},
  \bibinfo{journal}{J. Phys. A} \textbf{\bibinfo{volume}{41}},
  \bibinfo{pages}{065201} (\bibinfo{year}{2008}{\natexlab{b}}).

\bibitem[{\citenamefont{Glendinning}(1994)}]{94Glendinning}
\bibinfo{author}{\bibfnamefont{P.}~\bibnamefont{Glendinning}},
  \emph{\bibinfo{title}{Stability, Instability and Chaos}}
  (\bibinfo{publisher}{Cambridge University Press},
  \bibinfo{address}{Cambridge, UK}, \bibinfo{year}{1994}).

\bibitem[{\citenamefont{Mancini and Bose}(2004)}]{PhysRevA.70.022307}
\bibinfo{author}{\bibfnamefont{S.}~\bibnamefont{Mancini}} \bibnamefont{and}
  \bibinfo{author}{\bibfnamefont{S.}~\bibnamefont{Bose}},
  \bibinfo{journal}{Phys. Rev. A} \textbf{\bibinfo{volume}{70}},
  \bibinfo{pages}{022307} (\bibinfo{year}{2004}).

\bibitem[{\citenamefont{Wang and
  Schirmer}(2009{\natexlab{b}})}]{PhysRevA.80.042305}
\bibinfo{author}{\bibfnamefont{X.}~\bibnamefont{Wang}} \bibnamefont{and}
  \bibinfo{author}{\bibfnamefont{S.~G.} \bibnamefont{Schirmer}},
  \bibinfo{journal}{Phys. Rev. A} \textbf{\bibinfo{volume}{80}},
  \bibinfo{pages}{042305} (\bibinfo{year}{2009}{\natexlab{b}}).

\end{thebibliography}
\end{document}